\documentclass{article}
\usepackage[utf8]{inputenc}
\usepackage[T1]{fontenc}
\usepackage{makeidx}  
\usepackage{todonotes}
\usepackage{algorithm}
\usepackage{algorithmic}
\usepackage{multicol}
\usepackage{amsmath}
\usepackage{amssymb}

\newcommand{\constant}{\gamma}

\newcommand{\state}{P}
\newcommand{\statePhase}{P^r}

\newcommand{\CP}{\mathcal{P}}
\newcommand{\ord}{\tau}

\newcommand{\potentialDPhase}[1]{\Phi_{D_r}\left( #1 \right)}
\newcommand{\potentialAll}[1]{\Phi^N\left( #1 \right)}
\newcommand{\potentialFixed}[1]{\Phi^{N\backslash D_r}\left( #1 \right)}

\newcommand{\costshare}[1]{\chi_{ie}\left( #1 \right)}

\newcommand{\weightsSum}[1]{f_e\left(#1\right)}
\newcommand{\weightsSumSet}[2]{f^{#1}_e\left(#2\right)}

\newcommand{\movingPlayersPhase}{D_r}
\newcommand{\allPlayers}{N}

\newcommand{\fixedBPlayers}{\allPlayers \backslash B}

\newcommand{\playerCosts}[1]{X_i\left( #1 \right)}
\newcommand{\playerCostsMax}{X_{\text{max}}}
\newcommand{\playerCostsMin}{X_{\text{min}}}

\newcommand{\bestresponse}[1]{\mathcal{BR}_i\left( #1 \right)}
\newcommand{\sMove}{$s$-move}
\newcommand{\tMove}{$t$-move}

\newcommand{\tStretch}{t\text{-}\Omega_D}

\usepackage{amsthm}
\theoremstyle{plain}
\newtheorem{theorem}{Theorem}
\newtheorem{lemma}[theorem]{Lemma}
\newtheorem{corollary}[theorem]{Corollary}
\newtheorem{proposition}[theorem]{Proposition}
\newtheorem{claim}[theorem]{Claim}
\usepackage{authblk}

\begin{document}
\pagestyle{headings}  

\title{Computing Approximate Pure Nash Equilibria in Shapley Value Weighted Congestion Games
	\thanks{This work was partially supported by the German Research Foundation (DFG) within the Collaborative Research Centre ``On-The-Fly Computing'' (SFB 901) and by EPSRC grant EP/L011018/1.\newline\newline The final publication is available at Springer via http://dx.doi.org/10.1007/978-3-319-71924-5\_14.}}

%
%
\author[1]{Matthias Feldotto}
\author[2]{Martin Gairing}
\author[3]{Grammateia Kotsialou}
\author[1]{Alexander Skopalik}
%
%
%
\affil[1]{Paderborn University, Paderborn, Germany}
\affil[ ]{\{feldi,skopalik\}@mail.upb.de}
\affil[2]{University of Liverpool, Liverpool, UK}
\affil[ ]{m.gairing@liverpool.ac.uk}
\affil[3]{King's College London, London, UK}
\affil[ ]{grammateia.kotsialou@kcl.ac.uk}

\date{}

\maketitle              

\begin{abstract}
	We study the  computation of approximate pure Nash equilibria in  Shapley value (SV) weighted congestion games, introduced in \cite{DBLP:journals/teco/KolliasR15}. This class of games considers weighted congestion games in which Shapley values are used  as an alternative (to proportional shares) for distributing the total cost of each resource among its users. We focus on the interesting subclass of such games with polynomial resource cost functions and  present an algorithm that computes approximate pure Nash equilibria with a polynomial number of strategy updates. Since computing a single strategy update is hard, we apply sampling techniques which allow us to achieve polynomial running time. The algorithm builds on the algorithmic ideas of \cite{DBLP:journals/teco/CaragiannisFGS15}, however, to the best of our knowledge, this is the first algorithmic result on computation of approximate equilibria using  other than proportional shares as player costs  in this setting. We present a novel relation that approximates the Shapley value of a player by her proportional share and vice versa. As side results, we upper bound the approximate price of anarchy of such games and  significantly improve 
	the best known factor for computing approximate pure Nash equilibria in weighted congestion games  of \cite{DBLP:journals/teco/CaragiannisFGS15}.
\end{abstract}
\section{Introduction}
In many applications the state of a system depends on the behavior of individual participants
that act selfishly in order to minimize their own private cost. Non-cooperative game theory uses the concept of Nash equilibria as a tool for the theoretical analysis of such systems.  
A Nash equilibrium is a state in which no participant has an incentive to deviate to another strategy.
While mixed Nash equilibria, i.e., Nash equilibria in randomized strategies, are guaranteed to
exist under mild assumptions on the players’ strategy spaces and the private cost functions
 they are often hard to interpret. As a consequence, attention
is often restricted to pure Nash equilibria, i.e., Nash equilibria in deterministic strategies.

Rosenthal~\cite{RO73}  introduced a class of games, called
\emph{congestion games} that models a variety of  strategic interactions and is guaranteed to have pure Nash equilibria. 
In a congestion game, we are given a finite set of players $N$ and a finite set of resources $E$. A strategy
 of each player $i$ is to choose a subset of the resources out of a set $\CP_i$ of subsets of resources
allowable to her. In each strategy profile, each player pays for all used resources where the
cost of a resource $e\in E$ is a function $c_e$ of the number of players using it.  
Rosenthal used an elegant potential function argument to show that iterative improvement steps by the players converge to a pure Nash equilibrium and hence its existence is guaranteed.

Note that in congestion games each player using a resource has the same influence on the cost of this resource. 
To alleviate this limitation, \cite{MI96} and~\cite{DBLP:journals/tcs/FotakisKS05} studied a natural generalization called \emph{weighted congestion games} in which each player $i$ has a weight $w_i$  and the \emph{joint cost} 
of the resource is $f_e\cdot c_e(f_e)$, where $f_e$ is the total weight of players using $e$.
The joint cost of resource $e$ has to be covered by the set of players $S_e$ using it, i.e.,
$\sum_{i\in S_e}\chi_{ie}=f_e\cdot c_e(f_e)$, where $\chi_{ie}$ is the cost share of player $i$ 
on resource $e$.
The  \emph{cost sharing method} of the game defines
how exactly the joint cost of a resource is divided into individual cost shares $\chi_{ie}$. 
For weighted congestion games, the most widely studied cost sharing method is 
\emph{proportional sharing} (PS), where the cost share 
of a player is proportional to her weight, i.e., $\chi_{ie}=w_i \cdot c_e(f_e)$. 
Unfortunately,  weighted congestion games with proportional  sharing in general do not admit a 
\emph{pure} Nash equilibrium (see \cite{DBLP:journals/mor/HarksK12} for a characterization).

Kollias and Roughgarden \cite{DBLP:journals/teco/KolliasR15} proposed to use the \emph{Shapley value} (SV)
for sharing the cost of a resource in weighted congestion games. 
In the \emph{Shapley cost-sharing} method, the cost share 
of a player on a resource is the average marginal cost increase caused by her over all permutations of the players.
Using the Shapley value restores the existence of a potential function and therefore the existence of
pure Nash equilibria to such games \cite{DBLP:journals/teco/KolliasR15}. 
   
Potential functions immediately give rise to a simple and natural search procedure to find an equilibrium by performing iterative improvement steps starting from an arbitrary state. Unfortunately, this process may take exponentially many steps, even in the simple case of unweighted congestion games\footnote{Note that in the unweighted case, proportional sharing and Shapley cost sharing coincide.} and linear cost functions \cite{DBLP:journals/jacm/AckermannRV08}.
Moreover, computing  a pure Nash equilibrium in these games is intractable as the problem is  PLS-complete \cite{DBLP:conf/stoc/FabrikantPT04}, even for affine linear cost functions \cite{DBLP:journals/jacm/AckermannRV08}.
This result directly carries over to our game class with Shapley cost-sharing.
Given these intractability results, it is natural to ask for approximation which is formally captured by the concept of an
$\rho$-approximate pure Nash equilibrium.  This is a state from which no player can improve her  cost by a
factor of $\rho \ge 1$.  Recently,  Caragiannis et al. \cite{DBLP:conf/focs/CaragiannisFGS11} provided an algorithm to compute 
$\rho$-approximate Nash equilibria for unweighted congestion games
under proportional sharing. They also generalised their technique to
weighted congestion games \cite{DBLP:journals/teco/CaragiannisFGS15}.

\subsection{Our Contributions}

We present an algorithm to compute  $\rho$-approximate Nash equilibria in weighted congestion games under Shapley cost sharing. In games with polynomial cost functions of degree at most $d$, our algorithm 
achieves an approximation factor asymptotically  close to $\left(\frac{d}{\ln2} \right)^d \cdot poly(d)$.
Similar to \cite{DBLP:journals/teco/CaragiannisFGS15}
our algorithm computes a sequence of improvement steps of polynomial length that yields 
a $\rho$-approximate Nash equilibrium. Hence, our algorithm performs only a polynomial number of strategy updates. 
We show that our algorithm can also be used to compute $\rho$-approximate pure Nash equilibria for weighted congestion games with proportional sharing which improves the approximation factor of $d^{2\cdot d + o(d)}$  in \cite{DBLP:journals/teco/CaragiannisFGS15} to $\left(\frac{d}{\ln2} \right)^d \cdot poly(d)$.

We note that our method does not immediately yield an algorithm with polynomial running time since computing the Shapley cost share of a player and hence an improvement step is computationally hard.
However, we show that there is a polynomial-time randomized approximation scheme that can be used instead. This results in a randomized polynomial time algorithm that computes a strategy profile that is an approximate pure Nash equilibrium with high probability.

In the course of the analysis we exhibit an interesting relation between the Shapley cost share of a player and her proportional share. In the case of polynomial cost functions with constant degree, each of them can be approximated by the other within a constant factor.  This insight leads to an alternative proof to \cite{DBLP:conf/approx/HansknechtKS14} for the existence of approximate pure Nash equilibria in weighted congestion games with proportional cost sharing. 
  
Finally, we derive bounds on the approximate Price of Anarchy which may be of independent interest as they allow to bound the inefficiency of approximately stable states.

\subsection{Further Related Work}
 
Congestion games have been introduced by Rosenthal~\cite{RO73} who proved the existence of pure Nash equilibria by an exact potential function. Games admitting a potential function are called potential games and  each potential game is isomorphic to a congestion game~\cite{MS96}. 
Weighted congestion games were introduced by Milchtaich~\cite{MI96} and studied by Fotakis et al.~\cite{DBLP:journals/tcs/FotakisKS05}. Based on the Shapley value \cite{10.2307/1911054}, the class of weighted congestion games using Shapley values (instead of proportional shares) was introduced by \cite{DBLP:journals/teco/KolliasR15} and it was shown that such games are potential games. \cite{DBLP:journals/mor/GopalakrishnanMW14} extends this result by proving that a weighted generalisation of Shapley values is the only method that guarantee pure Nash equilibria. In contrast, proportional sharing does not guarantee existence of equilibria in general~\cite{DBLP:journals/mor/HarksK12}. 
Further research focuses on the quality of equilibria, measured by the Price of Anarchy (PoA)~\cite{DBLP:conf/stacs/KoutsoupiasP99}. For proportional sharing, Aland et al.~\cite{DBLP:conf/stacs/AlandDGMS06} show tight bounds on the PoA. Gkatzelis et al.~\cite{DBLP:conf/wine/GkatzelisKR14} show that, among all cost-sharing methods that guarantee existence of pure Nash equilibria, Shapley values minimise the worst PoA. Furthermore, tight bounds on PoA for general cost-sharing methods were given~\cite{DBLP:conf/icalp/GairingKK15}. For the extended model with non-anonymous costs by using set functions it was also shown that Shapley cost-sharing is the best method and tight results are given~\cite{DBLP:conf/ciac/KlimmS15,DBLP:journals/teco/RoughgardenS16}.

Computing a pure Nash equilibrium for congestion games was shown to be PLS-complete~\cite{DBLP:conf/stoc/FabrikantPT04} even for games with linear cost function~\cite{DBLP:journals/jacm/AckermannRV08}  or games with only three players~\cite{DBLP:journals/im/AckermannS08}.
Chien and Sinclair~\cite{DBLP:journals/geb/ChienS11} study the convergence towards $(1+\epsilon)$-approximate pure Nash equilibria in symmetric congestion games in polynomial time under a mild assumption on the cost functions.  In contrast, Skopalik and Vöcking show that this result cannot be generalized to asymmetric games and that computing a $\rho$-approximate pure Nash equilibrium is PLS-hard in general~\cite{DBLP:conf/stoc/SkopalikV08}.
Caragiannis et al.~\cite{DBLP:conf/focs/CaragiannisFGS11} give an algorithm which computes an $(2+\epsilon)$-approximate equilibrium for linear cost functions and an $d^{O(d)}$-approximate equilibrium for polynomial cost functions with degree of $d$. 
Weighted congestion games with proportional sharing do not posses pure Nash equilibria in general~\cite{DBLP:journals/tcs/FotakisKS05}. However, the existence of $d+1$-approximate equilibria for polynomial cost functions and $\frac{3}{2}$-approximate equilibria for concave cost functions was shown~\cite{DBLP:conf/approx/HansknechtKS14} and
Caragiannis et al.~\cite{DBLP:journals/teco/CaragiannisFGS15} present an algorithm for weighted congestion games and proportional sharing that computes $\frac{3+\sqrt{5}}{2}+\epsilon$-approximate equilibria for linear cost functions and $d^{2d+o(d)}$-approximate equilibria for polynomial cost functions.  

The computation of approximate equilibria requires the computation of Shapley values. In general, the exact computation is too complex. Mann and Shapley~\cite{mann1962values} suggest a sampling algorithm which was later analyzed by Bachrach et al.~\cite{DBLP:journals/aamas/BachrachMRPRS10} for simple coalitional games and by Aziz and de Keijzer~\cite{DBLP:conf/stacs/AzizK14} for matching games. Finally, Liben-Nowell et al.~\cite{DBLP:conf/cocoon/Liben-NowellSWW12} and  Maleki~\cite{maleki2015addressing} consider cooperative games with supermodular functions which correspond to our class.

\section{Our Model}\label{section:model}
A weighted congestion game is defined as $\mathcal{G} = (N, E, \left(w_i\right)_{i \in N}, \left(\CP_i\right)_{i \in N},\left(c_e\right)_{e \in E})$, where $N$ is the set of players, $E$ the set of resources, $w_i$ is the positive weight of player $i$, $\CP_i\subseteq 2^E$ the strategy set of player $i$ and  $c_e$ the cost function of resource $e$ (drawn from a set $\mathcal{C}$ of allowable cost functions). In this work, $\mathcal{C}$  is the set of polynomial functions with maximum degree $d$ and non-negative coefficients. The set of outcomes of this game is given by $\CP = \CP_1 \times \cdots \times \CP_n$, for an outcome, we write~$P=\left(P_1, \ldots, P_n\right) \in \CP$, where~$P_i\in \CP_i$. Let $(P_{-i}, P_i')$ be the outcome that results when player $i$ changes her strategy from $P_i$ to $P_i'$ and let $(P_{A}, P'_{N\setminus A})$ be the outcome that results when players $i\in A$ play their strategies in $P$ and players $i\in N\setminus A$ the strategies in $P'$.
The set of users of resource $e$ is defined by  $S_e(P)=\{i:e\in P_i\}$ and  the total weight on  $e$  by $f_e(P)=\sum_{i\in S_e(P)}w_i$. Furthermore, let $S_e^A(P)=\{i \in A:e\in P_i\}$ and $f^A_e(P)=\sum_{i\in S^A_e(P)}w_i$ be variants of these definitions with a restricted player set $A\subseteq N$.
The Shapley cost of a player $i$ on a resource $e$ is given as a function of the player's identity, the resource's cost function and her users $A$, i.e., $\chi_e (i, A)$. For simplicity, let $\chi_{ie}(P) = \chi_e (i,  S_e(P))$ be an abbreviation if all players are considered in a state $P$. Let $C_e(x) = x\cdot c_e(x)$. Then, the \emph{joint cost} on a resource $e$ is given by~$C_e(f_e(P)) = f_e(P)\cdot c_e(f_e(P))$ and the costs of players are such that $C_e(f_e(P)) =\sum_{i\in S_e(P)}\chi_{ie}(P)$. The \emph{ total cost} of a player~$i$ equals the sum of her costs in the resources she uses, i.e. $X_i(P) = \sum_{e\in P_i} \chi_{ie}(P)$. The social cost of the game is given by 
$SC(P) = \sum_{e \in E} f_e(P) \cdot c_e(f_e(P)) =  \sum_{e \in E} \sum_{i\in S_e(P)}\chi_{ie}(P)= \sum_{i  \in N}  X_i(P). $
Further define the social costs of a subset of players $A \subseteq N$ with $SC_A(P) = \sum_{i  \in A} X_i(P)$.
	
The  cost-sharing method is important for our analysis, as it defines how the joint cost on a resource $e$ is distributed among her users. In this paper, the methods we focus on are the Shapley value and the proportional cost-sharing, which we introduce in detail.

\smallskip\noindent \textbf{Shapley values.} \label{sv-def} 
For a set of players $A$, let $\Pi(A)$ be the set of permutations $\pi: A \rightarrow A\left\{1, \ldots, |A|\right\}$. For a $\pi \in \Pi(A)$, define as $A^{<i,\pi} = \left\{j \in A: \pi(j) < \pi(i)\right\}$  the set of players preceding player $i$ in $\pi$ and as $W_A^{<i,\pi} = \sum_{j \in A:\pi(j) < \pi(i)} w_j$ the sum of their weights.

For the uniform distribution over $\Pi(A)$, the  Shapley value of a player $i$ on resource $e$  is  given by
\vspace*{-.5em}
\begin{align*}
\chi_e (i, A)= E_{\pi\sim\Pi(A)}\left[C_e \left(W_A^{<i,\pi} +w_i\right) -  C_e\left(W_A^{<i,\pi}\right)\right].
\end{align*}

\smallskip\noindent\textbf{Proportional sharing.} \label{prop-def}  The cost of a player $i$ on a resource under proportional sharing is given by
$\chi_{ie}^{\text{Prop}}(P) = w_i \cdot c_e(f_e(P))$.
For the rest of the paper,  we write $X_i^{\text{Prop}}(P)= \sum_{e \in E} \chi_{ie}^{\text{Prop}}(P)$ to indicate when we switch to proportional sharing.
	
\smallskip\noindent\textbf{$\rho$-approximate pure Nash equilibrium.} Given a parameter $\rho \geq 1$ and an outcome $P$, we call as \emph{$\rho$-move} a deviation from $P_i$ to $P'_i$ where the player improves her cost by more than a factor $\rho$, formally $X_i(P) > \rho \cdot X_{i}(P_{-i}, P_i')$. We call the state $P$ an {\em $\rho$-approximate pure Nash equilibrium} ($\rho$-PNE) if and only if no player is able to perform a $\rho$-move, formally it holds for every player~$i$ and any other strategy $P_i'\in\CP_i$ that $X_i(P)\le \rho \cdot X_{i}(P_{-i}, P_i')$.
	
\smallskip\noindent\textbf{$\rho$-approximate Price of Anarchy.}
Given a  parameter $\rho \ge 1$, let~$\rho\text{-PNE} \subseteq \CP$ be the set of $\rho$-approximate pure Nash equilibria and $P^*$ the state of optimum, i.e., $P^*=\min_{P'\in\CP}SC(P')$. Then the {\em $\rho$-approximate price of anarchy} ($\rho$-PoA)
is defined as $\rho\text{-PoA}=\max_{P\in\rho\text{-PNE}}\frac{SC(P)}{SC(P^*)}$.

Kollias and Roughgarden~\cite{DBLP:journals/teco/KolliasR15} prove that weighted congestion games under Shapley values are potential games using the following potential.
	
\smallskip\noindent	\textbf{Potential Function.}
Given an outcome $P$ and an arbitrary ordering $\ord$ of the players in $N$, the potential is given by 
\begin{align} \label{pot}
	\Phi(P) = \sum_{e\in E} \Phi_e(P) = \sum_{e\in E} \underset{i \in S_e(P)}{\sum} \chi_{e}(i, \{j:\ord(j)\leq \ord(i), j\in S_e(P)\}).
\end{align}

\smallskip\noindent\textbf{$A$-limited potential.} We now restrict this  potential function by  allowing only a subset of players $A\subseteq N$ to participate and define the $A$-limited potential as
\begin{align} \label{lpot}
	\Phi^A(P) = \sum_{e \in E} \Phi^A_e(P) = \sum_{e\in E} \underset{i \in S_e^A(P)}{\sum} \chi_{e}(i, \{j:\ord(j)\leq \ord(i), j\in S_e^A(P)\}).
\end{align}

\smallskip\noindent \textbf{$B$-partial potential.} Consider sets $A$ and $B$ such that $B \subseteq A \subseteq N$. Then the $B$-partial potential of set $A$ is defined by
\begin{align} \label{ppot}
	\Phi^A_B(P)= \Phi^A(P) - \Phi^{A\backslash B}(P) = \sum_{e \in E} \Phi^A_{e,B}(P) = \sum_{e \in E} \Phi_e^A(P) - \Phi_e^{A\backslash B}(P).
\end{align}
If the set $B$ contains only one player, i.e., $B=\{\{i\}\}$, then we write $\Phi^A_i(P)=\Phi^A_B(P)$. In case of $A=N$, $\Phi^N_B(P) = \Phi_B(P)= \sum_{e \in E} \Phi_{e,B}(P)$.
Intuitively, $\Phi^A_B(P)$ is the value that the players in $B \subseteq A$ contribute to  the $A$-limited potential.

\smallskip\noindent	\textbf{$\rho$-stretch.} Similar to  $\rho$-PoA, we define a ratio with respect to the potential function. Let $\hat{P}$ be the outcome that minimises the potential, i.e., $\hat{P}=\min_{P'\in\CP}\Phi(P')$. Then 
the {\em $\rho$-stretch} is defined as 
\begin{align}
\rho\text{-}\Omega= \max_{P\in\rho\text{-PNE}}\frac{\Phi(P)}{\Phi(\hat{P})}.
\end{align}
\smallskip\noindent	\textbf{$A$-limited $\rho$-stretch.} Additionally, we define a $\rho$-stretch restricted to players in a subset $A\subseteq N$. Let~$\rho\text{-PNE}_A \subseteq \CP$ be the set of $\rho$-approximate pure Nash equilibria where only players in $A$ participate. The rest of the players have a fixed strategy $\bar{P}_{N\setminus A}$. Then we define the $A$-limited $\rho$-stretch as
\begin{align}
\rho\text{-}\Omega_A=\max_{P\in\rho\text{-PNE}_A}\frac{\Phi(P)}{\Phi(\hat{P})} = \max_{P\in\rho\text{-PNE}_A} \frac{\Phi(P_A, \bar{P}_{N\setminus A})}{\Phi(\hat{P}_A, \bar{P}_{N \setminus A})}.
\end{align}

\section{Algorithmic Approach and Outline}

Our algorithm is based on ideas by Caragiannis et al.~\cite{DBLP:journals/teco/CaragiannisFGS15}. 
Intuitively, we partition the players' costs into intervals  $[b_1, b_2],[b_2,b_3] , \ldots ,[b_{m-1},b_m]$  in decreasing order. The cost values in one interval are within a polynomial factor. 
Note that this ensures that every sequence of $\rho$-moves for $\rho>1$ of players with costs in one or two intervals converges in polynomial time.
 
After an initialization, the algorithm proceeds in phases $r$ from $1$ to $m-1$.
In each phase $r$, players with costs in the interval $[b_r, + \infty]$ do $\alpha$-approximate moves where  $\alpha$ is close to the desired approximation factor.
Players with costs in the interval $[b_{r+1}, b_r]$ make $1+\gamma$-moves for some small $\gamma > 0$.   
After a polynomial number of steps no such moves are possible and we freeze all players with costs in  $ [b_r, + \infty]$.  These players will never be allowed to move again. We then proceed with the next phase. Note that at the time players are frozen, they are in an  $\alpha$-approximate equilibrium.  The purpose of the $1+\gamma$-moves of players of the neighboring interval is to ensure that the costs of frozen players do not change
significantly in later phases. To that end we utilize a potential function argument.  We argue about the potential of \emph{sub games} among a subset of players. We can bound the potential value of an arbitrary $q$-approximate equilibrium with the minimal potential value (using the \textit{stretch}). Compared to the approach in~\cite{DBLP:journals/teco/CaragiannisFGS15}, we directly work with the exact potential function of the game which significantly improves the results, but also requires a more involved analysis. We show that the potential of the sub game in one phase is significantly smaller than $b_r$. Therefore, the costs experienced by players moving in phase $r$ are considerably lower than the costs of any player in the interval $[b_1, b_{r-1}]$. 
The analysis heavily depends on the stretch of the potential function which we  analyze in Section~\ref{section:PoA_Stretch}. The proof there is based on the technique of Section~\ref{section:Approximation} in which we approximate the Shapley with proportional cost sharing. For the technical details in both sections we need some structural properties of costs-shares and the restricted potentials which we show in the next section.

\section{Shapley and Potential Properties}\label{section:preliminary_results}
	
The following properties of the Shapley values are extensively used in our proofs.
\begin{proposition} \label{prop:properties}
 Fix a resource $e$. Then
	for any set of players $S$ and $i\in S$, we have for
	$j,j_1,j_2,j',j'_1,j'_2,i_1,i_2\not\in S$:
	\begin{enumerate}\renewcommand{\theenumi}{\alph{enumi}}
		\item $\chi_{e}\left(i, S\right) \leq \chi_{e}\left(i, S \cup \{j\}\right) $,
		\vspace{1mm}
		\item $\chi_{e}\left(i, S\cup\{j'\}\right) \geq \chi_{e}\left(i, S \cup \{j_1, j_2\}\right)$,  with $j' \neq i$ and $w_{j'} = w_{j_1} + w_{j_2}$, 
		\vspace{1mm}			
		\item $\chi_{e}\left(i, S \cup \{j_1, j_2\} \right) \geq \chi_{e}\left(i, S \cup \{j'_1, j'_2\}\right)$, with $w_{j'_1}=w_{j'_2} = \frac{w_{j_1} + w_{j_2}}{2}$,
		\vspace{1mm}
		\item $\chi_{e}\left(i, S\right) \geq \chi_{e}\left(i_1, S\backslash \{i\} \cup \{i_1\}\right)  +\chi_{e}\left(i_2, S\backslash \{i\} \cup \{i_1, i_2\}\right)$, with $w_{i_1}=w_{i_2} = \frac{w_i}{2}$.
	\end{enumerate}
\end{proposition}
	
\noindent
We proceed to the properties of the restricted types of potential defined before.

\begin{proposition}\label{prop:potentials}
	Let $A$ and $B$ be sets of players such that $B \subseteq A \subseteq N$, $P$ and $P'$  outcomes of the game such that the players in $A \subseteq N$ use the same strategies in both $P$ and $P'$, and  $z \in N$ an arbitrary player. Then
	\begin{multicols}{3}
	\begin{enumerate}\renewcommand{\theenumi}{\alph{enumi}}
		\item $\Phi_B^A(P)\leq \Phi_B(P)$,\label{prop:partialtofull}
		\item $\Phi_B^A(P)= \Phi_B^A(P')$,\label{prop:sameStrategies}
		\item  $\Phi_z(P)= X_z(P)$.\label{prop:potentialcosts} 
	\end{enumerate}
	\end{multicols}
\end{proposition}
Next, we show that the potential property also holds for the partial potential.
	
\begin{proposition}\label{prop:diffpotentialcosts}
	Consider a subset $B\subseteq N$ and a player $i \in B$. Given two states, $P$ and $P'$, that differ only in the strategy of player $i$, then~$\Phi_B(P) - \Phi_B(P') = X_i(P) - X_i(P')$.
\end{proposition}
The next lemma gives a relation between partial potential and  Shapley values.
\begin{lemma} \label{lemma:resource_potential_sum_costs}
	Given an outcome $P$ of the game, a resource $e$ and a subset $B \subseteq N$, it holds that~$\Phi_{e,B}(P) \le \sum_{i\in B}\chi_{ie}(P) \leq \Phi_{e,B}(P) \cdot (d+1)$.
\end{lemma}
Summing up over all resources $e \in E$ yields the next corollary.
	
\begin{corollary} \label{cor:potential_sum_costs}
	Given an outcome $P$ of the game and a subset $B \subseteq N$, it holds that~$\Phi_B(P)  \le \sum_{i\in B} X_i(P) \leq \Phi_B(P) \cdot (d+1)$.
\end{corollary}

\section{Approximating Shapley with Proportional Cost-Shares}
\label{section:Approximation}

In this section we approximate the Shapley value of a player with her proportional share. This approximation  plays an important role in our proofs of the stretch and for the computation.

\vspace{-1.5mm}
\begin{lemma} \label{lemma:shatoprop}
For a player $i$, a resource $e$ and any state $\state$, the following inequality holds between her Shapley and proportional cost:
\begin{align*}
\frac{2}{d+1}\cdot \chi_{ie}(P) &\leq  \chi_{ie}^{\text{Prop}}(P)  \leq\frac{d+3}{4} \cdot \chi_{ie}(P).
\end{align*}
\end{lemma}
	
Summing up over all $e \in E$ implies the following corollary.
\begin{corollary} \label{corollary:shatoprop}
	For a player $i$ and any state $\state$, the following inequality holds between her Shapley and proportional cost:
	\begin{align*}
		\frac{2}{d+1}\cdot X_i(P) &\leq  X_i^{\text{Prop}}(P)  \leq\frac{d+3}{4} \cdot X_i(P).
	\end{align*}
\end{corollary}

\begin{lemma} \label{lemma:approximate_equilibria}
	Any $\rho$-approximate pure Nash equilibrium for a SV weighted congestion game of degree $d$ is a $\frac{(d+3)\cdot(d+1)}{8}\cdot\rho$-approximate pure Nash equilibrium for the weighted congestion game with proportional sharing.
\end{lemma}

\section{The Approximate Price of Anarchy and Stretch}
\label{section:PoA_Stretch}

Firstly, we upper bound the approximate Price of Anarchy for our game class.
\begin{lemma}\label{lemma:approx-PoA}
	Let $\rho \geq 1$ and $d$ the maximum degree of the polynomial cost functions. Then
	\[\rho\text{-PoA} \leq \frac{\rho \cdot (2^{\frac{1}{d+1}} - 1)^{-d}}{2^{-\frac{d}{d+1}} \cdot( 1+ \rho) - \rho}.\]
\end{lemma}

Similar to the $\rho$-PoA, we  also derive an upper bound on the $\rho$-stretch which expresses the ratio between local and global optimum of the potential function.
\begin{lemma}\label{lemma:approx-stretch}
	Let $\rho \geq 1$ and $d$ the maximum  degree of the polynomial cost functions. Then an upper bound for the $\rho$-stretch of polynomial  SV weighted congestion games  
 is
	\[\rho\text{-}\Omega \leq \frac{\rho \cdot (2^{\frac{1}{d+1}} - 1)^{-d}\cdot (d+1)}{2^{-\frac{d}{d+1}} \cdot( 1+ \rho) - \rho}.\]
\end{lemma}
We now proceed to  the upper bound  of the $D$-limited $\rho$-stretch. To do this, we use the $\rho$-PoA (Lemma \ref{lemma:approx-PoA}) and  Lemmas \ref{lemma:relation-social-costs} and \ref{lemma:approx-limited-stretch-social-costs}, which we prove next. 
\begin{lemma}\label{lemma:relation-social-costs}
	Let $\rho \geq 1$, $d$ the maximum degree of the polynomial cost functions and $\hat{P}=\min_{P'\in\CP}\Phi(P')$. Then
	\[
	\frac{SC(P)}{SC(\hat{P})}
	\leq
	\frac{\rho \cdot (2^{\frac{1}{d+1}} - 1)^{-d}}{2^{-\frac{d}{d+1}} \cdot( 1+ \rho) - \rho}.
	\]
\end{lemma}

\begin{proof}
	Let $P$ be an $\rho$-approximate equilibrium and $P^*$ the optimal outcome.  Let $\hat{P}=\min_{P'\in\CP}\Phi(P')$ be the minimizer of the potential and by definition also a pure Nash equilibrium. Then we can lower bound the $\rho$-PoA as follows,
	\begin{align}
	\rho\text{-PoA} = \underset{P \in \rho\text{-PNE}}{\max} \frac{SC(P)}{SC(P^*)}
	\ge \underset{P \in \rho\text{-PNE}}{\max} \frac{SC(P)}{SC(\hat{P})}. \label{lower}
	\end{align}
Lemma~\ref{lemma:approx-PoA} and \eqref{lower} give that $
	\underset{P \in \rho\text{-PNE}}{\max} \frac{SC(P)}{SC(\hat{P})}
	\leq \rho\text{-PoA}
	\leq \frac{\rho \cdot (2^{\frac{1}{d+1}} - 1)^{-d}}{2^{\frac{-d}{d+1}} \cdot( 1+ \rho) - \rho}.
	$ \hfill 
\end{proof}

\begin{lemma}\label{lemma:approx-limited-stretch-social-costs}
	Let $\rho \geq 1$, $d$ the maximum degree of the polynomial cost functions and $D\subseteq N$ an arbitrary subset of players. Then
	\[
	\rho\text{-}\Omega_D
	\leq
	\frac{(d+1)^2 \cdot (d+3)}{8} \cdot \frac{SC(P)}{SC(\hat{P})}.
	\]
\end{lemma}

By  Lemma~\ref{lemma:relation-social-costs} and Lemma~\ref{lemma:approx-limited-stretch-social-costs}, we get the following desirable  corollary.
\begin{corollary} \label{stretch:upperbound}
	For $\rho \geq 1$, $d$ the maximum degree of the polynomial cost functions and $D\subseteq N$ an arbitrary subset of players,
	\[\rho\text{-}\Omega_D \leq \frac{(d+1)^2 \cdot (d+3)}{8} \cdot \frac{\rho \cdot (2^{\frac{1}{d+1}} - 1)^{-d}}{2^{-\frac{d}{d+1}} \cdot( 1+ \rho) - \rho}. \]
\end{corollary}

\section{Computation of Approximate Pure Nash Equilibria}\label{section:computation}

To compute $\rho$-approximate pure Nash equilibria in SV congestion  games, we construct an algorithm based on the  idea by Caragiannis et al.~\cite{DBLP:journals/teco/CaragiannisFGS15}. The main idea is to separate the players in different blocks depending on their costs.  The players who are processed first are the ones with the largest costs followed by the smaller ones. The size of the blocks and the distance between them  is polynomially bounded by the number of players $n$ and  the maximum degree $d$ of the polynomial cost functions $c_e$. Formally, we define $\playerCostsMax = \max_{i \in \allPlayers}\playerCosts{\state}$ as  the maximum cost among all players before running the algorithm. Let $\bestresponse{0}$ be a state of the game in which only player $i$ participates and plays her best move. Then, define as $\playerCostsMin = \min_{i \in \allPlayers}\playerCosts{\bestresponse{0}}$  the minimum possible cost in the game. Let $\constant$ be an arbitrary constant such that $\gamma >0$, $m = \log\left(\frac{\playerCostsMax}{\playerCostsMin}\right)$ is the number of different blocks and  $b_r=  \playerCostsMax \cdot g^{-r}$ the block size for any $r \in [0, m]$, where $g=2 \cdot n \cdot (d+1) \cdot \constant^{-3}$.
	
The algorithm is now executed in $m-1$ phases. Let $\state$ be the current state of the game and, for each phase $r \in [1, m-1]$, let $\state^r$ be the state before phase $r$. All players $i$ with $\playerCosts{P} \in [b_r, +\infty]$  perform an \sMove \ with $s=\left(\frac{1}{\tStretch}-2\constant\right)^{-1}$ (almost $\tStretch$-approximate moves), while all players $i$ with $\playerCosts{P} \in [b_{r+1}, b_r]$ perform a \tMove \ with $t=1+\gamma$ (almost pure moves).
Let $\bestresponse{\state}$ be the best response of player $i$ in state $\state$. The phase ends when the first and the second group of  players  are in an $s$- and $t$-approximate equilibrium, respectively. At the end of the phase,  players with $\playerCosts{P}>b_r$ have  irrevocably decided their strategy and have been added  in the list of finished players. In addition, before  the described phases are executed, there is an initial phase in which all players with $\playerCosts{P}\geq b_1$ can perform a \tMove \ to prepare the first real phase.
	
	\begin{algorithm}
		\caption{Computation of approximate pure Nash equilibria}
		\label{algorithm:computation}
		\begin{algorithmic}
			\STATE $\playerCostsMax = \max_{i \in \allPlayers}\playerCosts{\state}$, $\playerCostsMin = \min_{i \in \allPlayers}\playerCosts{\bestresponse{0}}$, $m = \log\left(\frac{\playerCostsMax}{\playerCostsMin}\right)$
			\STATE $\gamma >0$, $g=2 \cdot n \cdot (d+1) \cdot \constant^{-3}$, $b_r=  \playerCostsMax \cdot g^{-r} \forall \in [0, m]$
			\STATE  $t=1+\gamma$, $s=\left(\frac{1}{\tStretch}-2\constant\right)^{-1}$
			\WHILE{there is a player $i\in \allPlayers$ with $\playerCosts{\state} \geq b_1$ and who can perform a $t$-move}
			\STATE $P \leftarrow \left(\state_{-i}, \bestresponse{\state}\right)$
			\ENDWHILE
			\FORALL{phases $r$ from $1$ to $m-1$}
			\WHILE{there is a non-finished player $i\in \allPlayers$ either with $\playerCosts{\state} \in [b_r, +\infty]$ and who can perform a $s$-move or with $\playerCosts{\state} \in [b_{r+1}, b_r]$ and who can perform a $t$-move}
			\STATE $\state \leftarrow \left(\state_{-i}, \bestresponse{\state}\right)$
			\ENDWHILE
			\STATE Add all players $i \in \allPlayers$ with $\playerCosts{\state} \geq b_r$ to the set of finished players.
			\ENDFOR
		\end{algorithmic}
	\end{algorithm}
	
For the analysis, let $\movingPlayersPhase$ be the set of deviating players in phase $r$ and $\state^{r,i}$ denote the state after player $i \in D_r$ has done her last move within phase $r$.

\vspace{-1mm}
\begin{theorem}\label{theorem:algorithm}
An $\alpha$-approximate pure Nash equilibrium with $\alpha \in \left(\frac{d}{\ln 2}\right)^d \cdot poly(d)$ can be computed with a polynomial number of improvement steps.
\end{theorem}
\begin{proof}
The main argument follows from bounding  the $D$-partial potential of the moving players in each phase (see  Lemma~\ref{lemma:potential_block}).	To that end, we first prove that  the partial potential is bounded by the sum of the  costs of players  when they did their last move (Lemma \ref{lemma:potential_costs}).

\begin{lemma}\label{lemma:potential_costs}
For every phase $r$, it holds that $\potentialDPhase{\statePhase} \leq \sum_{i \in \movingPlayersPhase} {\playerCosts{\state^{r,i}}}$.	
\end{lemma}
We now use the Lemma \ref{lemma:potential_costs} and the stretch of the previous section to bound the potential of the moving players by the according block size.
\begin{lemma}\label{lemma:potential_block}
For every phase $r$, it holds that $\potentialDPhase{\state^{r-1}} \leq \frac{n}{\constant} \cdot b_r$.
\end{lemma}
 It remains to show that  the running time is bounded and that the approximation factor holds. For the first, since the partial potential is bounded and each deviation decreases the potential, we can limit the number of possible improvement steps (see Lemma~\ref{lemma:running_time}).
\begin{lemma}\label{lemma:running_time}
The algorithm uses a polynomial number of improvement steps.
\end{lemma}
We show next that every player who has already finished his movements will not get much worst costs at the end of the algorithm (see Lemma~\ref{lemma:costs_end}) and that there is no alternative strategy which is more attractive at the end (see Lemma~\ref{lemma:costs_otherstrategy}).
\begin{lemma}\label{lemma:costs_end}
Let $i$ be a player who makes her last move in phase $r$ of the algorithm. Then,
			$\playerCosts{\state^{m-1}}\leq (1+\constant^2) \cdot \playerCosts{\statePhase}.$
\end{lemma}
\begin{lemma}\label{lemma:costs_otherstrategy}
Let $i$ be a player who makes her last move in phase $r$ and let $\state'_i$ be an arbitrary strategy of $i$. Then,
$\playerCosts{\state_{-i}^{m-1}, \state'_i} \geq (1- \constant) \cdot \playerCosts{\state_{-i}^r, \state'_i}.$
\end{lemma}
Next, we bound the approximation factor of the whole algorithm (see Lemma~\ref{lemma:approximation_factor}).
\begin{lemma}\label{lemma:approximation_factor}
After the last phase of the algorithm, every player $i$ is in an $\alpha$-approximate pure Nash equilibrium with $\alpha = (1+O(\constant)) \cdot \tStretch$.
\end{lemma}
		The polynomial running time and the approximation factor of $\alpha = (1+O(\gamma)) \cdot \tStretch$ follow directly from Lemma~\ref{lemma:running_time} and Lemma~\ref{lemma:approximation_factor}. Last, using  Corollary \ref{stretch:upperbound}, we show that  $\alpha \in \left(\frac{d}{\ln 2}\right)^d \cdot poly(d)$.
		\begin{lemma} \label{alpha_order}
 The approximation factor $\alpha$ is in the order  of $  \left(\frac{d}{\ln2}\right)^d \cdot poly(d)$. 
 \end{lemma}
This completes the proof of Theorem \ref{theorem:algorithm}.\hfill
\end{proof}
	
We note that a significant improvement below $O\left(\left(\frac{d}{\ln 2}\right)^d\right)$ of the approximation factor would require new algorithmic ideas as the lower bound of the PoA in \cite{DBLP:conf/wine/GairingS07} immediately yields a corresponding lower bound on the stretch.

This algorithm can be used to compute also approximate pure Nash equilibria in weighted congestion games (with proportional sharing). Such a game  can now be approximated  by a Shapley game losing only  a factor of $\frac{(d+3)(d+1)}{8}$ (by Lemma~\ref{lemma:approximate_equilibria}), which  is included in $poly(d)$.
\begin{corollary}
	For any weighted congestion game with proportional sharing, an $\alpha$-approximate pure Nash equilibrium with $\alpha \in \left(\frac{d}{\ln 2}\right)^d \cdot poly(d)$  can be computed with a polynomial number of improvement steps.
\end{corollary}

\subsection{Sampling Shapley Values}
\label{section:sampling}
	
The previous section gives an algorithm with polynomial running time with respect to the number of improvement steps. However, each improvement step requires the multiple computations of Shapley values, which are hard to compute. For this reason,  one can instead compute an approximated  Shapley value with sampling methods.
Since we are only interested in approximate equilibria, an execution of the algorithm with approximate steps has a negligible impact on the final result.
The technical properties of Shapley values stated in Section~\ref{section:preliminary_results} also hold for sampled instead of exact Shapley values with high probability.
\begin{theorem}\label{sampling}
For any constant $\gamma$, an $\alpha$-approximate pure Nash equilibrium with $\alpha \in \left(\frac{d}{\ln 2}\right)^d \cdot poly(d)$ can be computed in polynomial time with high probability.
	\end{theorem}
\begin{proof} We use sampling techniques that follow~\cite{DBLP:conf/cocoon/Liben-NowellSWW12,mann1962values} and adjust them to our setting.

\begin{algorithm}
\caption{Approximation of the Shapley value by sampling}
		\label{algorithm:sampling}
		\begin{algorithmic}
			\FORALL{$r$ from $1$ to $\log\left(2 n^{c+3} \cdot \max_{i \in N}\CP_i \cdot |E| \cdot \left(1+\log\left(\frac{\playerCostsMax}{\playerCostsMin}\right)\right) \cdot (d+1) \cdot \gamma^{-9}\right)$}
				\FORALL{$j$ from $1$ to $k=\frac{4(|S_e(P)|-1)}{\mu^2}$}
				\STATE Pick uniformly at random permutation $\pi$ of the players $S_e(P)$ using resource $e$
				\STATE Compute marginal contribution $MC_{ie}^j(P) = C_e \left(W_{S_e(P)}^{<i,\pi} +w_i\right) -  C_e\left(W_{S_e(P)}^{<i,\pi}\right)$
				\ENDFOR
				\STATE Let $\overline{MC}_{ie}(P) = \frac{1}{k} \sum_{j=1}^k MC_{ie}^j(P)$
			\ENDFOR
			\STATE Return the median of all $\overline{MC}_{ie}(P)$
		\end{algorithmic}
		\end{algorithm}

\begin{lemma}\label{lemma:sampling_shapley}
	Given an arbitrary state $P$ and an arbitrary but fixed constant $c$, Algorithm~\ref{algorithm:sampling}  computes a $\mu$-approximation of $\chi_{ie}(P)$ for any player $i$ in polynomial running time with probability at least $$1-\left(n^c \cdot n \cdot \max_{i \in N}\CP_i \cdot |E| \cdot \left(1+\log\left(\frac{\playerCostsMax}{\playerCostsMin}\right)\right) \cdot 2 \cdot n^2 \cdot (d+1) \cdot \gamma^{-9}\right)^{-1}.$$
\end{lemma}
For using the sampling in the computation of an improvement step, a Shapley value has to be approximated for each alternative strategy of a player and for each resource in the strategy. In the worst case, each player has to be checked for an available improvement step.

\begin{lemma}\label{lemma:sampling_improvement_step}
	Given an arbitrary state $P$  and running the sampling algorithm at most $n \cdot \max_{i \in N}\CP_i \cdot |E|$ times computes an improvement step for an arbitrary player with probability at least $1-\left(n^c \cdot \left(1+\log\left(\frac{\playerCostsMax}{\playerCostsMin}\right)\right) \cdot 2 n^2 \cdot (d+1) \cdot \gamma^{-9}\right)^{-1}$.
\end{lemma}
Lemma~\ref{lemma:running_time} gives a bound on the number of improvement steps. Using the sampling algorithm for $\mu=1+\gamma$, we can bound the total number of samplings:

\begin{lemma}\label{lemma:sampling_total}
	During the whole execution of Algorithm~\ref{algorithm:computation} the sampling algorithm for $\mu=1+\gamma$ is applied at most $n \cdot \max_{i \in N}\CP_i \cdot |E| \cdot \left(1+\log\left(\frac{\playerCostsMax}{\playerCostsMin}\right)\right) \cdot 2 \cdot n^2 \cdot (d+1) \cdot \gamma^{-9} $ times and the computation of the approximate pure Nash equilibrium is correct with probability at least $1-n^{-c}$ for an arbitrary constant c.
\end{lemma}
Summing up, we show that a $\mu$-approximation of one Shapley value can be computed in polynomial running time with high probability (Lemma \ref{lemma:sampling_shapley}) and the sampling algorithm is running at most a polynomial number of times (Lemma \ref{lemma:sampling_total}). Then Theorem \ref{sampling} follows.
\hfill
\end{proof}



\bibliographystyle{splncs03}
\bibliography{ApproximateShapley}


\newpage

\appendix

\section*{Appendix}

\section{Proofs for the Properties in Section~\ref{section:preliminary_results}}

\begin{proof}[Proof of Proposition~\ref{prop:properties}]

 Let $k:=|S|$. By the definition of Shapley values 
\begin{align*}
	\chi_e(i,S\cup\{j\})	
	&= \frac{1}{(k+1)!}\sum_{\pi\in\Pi(S\cup\{j\})} 
	\left(
	C_e\left( W_{S \cup \{j\}}^{<i, \pi}
	+w_i\right)
	- C_e\left( W_{S \cup \{j\}}^{<i, \pi}\right)  \right)\\
	&\ge \frac{1}{(k+1)!}\sum_{\pi\in\Pi(S\cup\{j\})} 
	\left(
	C_e\left( W_{S }^{<i, \pi} + w_i \right)
	- C_e\left( W_{S }^{<i, \pi}\right)
	\right)\\
	&= \frac{1}{k!}\sum_{\pi\in\Pi(S)} 
	\left(
	C_e\left( W_{S }^{<i, \pi} + w_i \right)
	- C_e\left( W_{S }^{<i, \pi} \right)
	\right)\\
	&= \chi_e(i,S),
\end{align*}
proving (a).
    
 For (b) and (c), consider $\chi_{e}\left(i, S \cup \{j_1, j_2\}\right)$. 
 Observe, that only for permutations $\pi\in\Pi(S\cup\{j_1,j_2\})$ where either $j_1< i < j_2$ or $j_2< i < j_1$ the corresponding contribution to $\chi_e\left(i, S \cup \{j_1, j_2\}\right)$ 
 changes if we change the weight of $j_1,j_2$ but keep their sum the same.    
 Fix a permutation $\pi\in\Pi(S\cup\{j_1,j_2\})$ with $j_1< i < j_2$ and pair it with the corresponding permutation $\hat{\pi}$ where only $j_1$ and $j_2$ are swapped.
 Then the contribution of $\pi$ and $\hat{\pi}$ to $\chi_{e}\left(i, S \cup \{j_1, j_2\}\right)$ is 
 \begin{align}
    \frac{1}{(k+2)!} \cdot
         &\left( 
               C_e\left(  W^{<i,\pi}_S + w_{j_1} +w_i \right)
    -C_e\left( W^{<i,\pi}_S + w_{j_1} \right)\right) \notag \\
             &\left.\quad +  C_e\left(  W^{<i,\pi}_S + w_{j_2} +w_i \right)
    - C_e\left( W^{<i,\pi}_S + w_{j_2} \right)
  \right).\label{eq:cscontribution}
 \end{align}
Since $C_e(x+w_i)-C_e(x)$ is convex in $x$, we get that 
\begin{align*}
\eqref{eq:cscontribution} \geq
\frac{1}{(k+2)!} \cdot
         &\left( 
               C_e\left( W^{<i,\pi}_S + w_{j'_1} +w_i\right)
    -C_e\left(W^{<i,\pi}_S + w_{j'_1} \right)\right. \notag \\
             &\left.\quad +  C_e\left(  W^{<i,\pi}_S + w_{j'_2} +w_i\right)
    -C_e\left( W^{<i,\pi}_S + w_{j'_2} \right)
  \right),
\end{align*}
and 
   \begin{align*}
\eqref{eq:cscontribution} \leq
\frac{1}{(k+2)!} \cdot
         &\left( 
               C_e\left(  W^{<i,\pi}_S + w_{j_1}+w_{j_2} +w_i\right)
    -C_e\left( W^{<i,\pi}_S +  w_{j_1}+w_{j_2}\right)\right. \notag \\
             &\left.\quad +  C_e\left( W^{<i,\pi}_S + 0  +w_i\right)
    -C_e\left( W^{<i,\pi}_S + 0\right)
  \right).
\end{align*}
Part (c) and (b) follow, respectively.
Part (d) of the proposition is shown in \cite{DBLP:conf/icalp/GairingKK15}.
\hfill 
\end{proof}
\vspace{3mm}

\begin{proof}[Proof of Proposition~\ref{prop:potentials}]
	We prove the different parts separately:
	\begin{enumerate}\renewcommand{\theenumi}{\alph{enumi}}
		\item 
For each $e\in E$, let  $I_e(P) = \Phi^A_e(P) - \Phi^{A\backslash B}_e(P)$. By definition of the $B$-partial potential \eqref{ppot}, we have 
\begin{align} \label{cl}
	\Phi^A_B(P) = \Phi^A(P) - \Phi^{A\backslash B}(P) = \underset{e\in E}{\sum} I_e(P).
\end{align}
By the definition of  limited potential \eqref{lpot}, for an arbitrary $\ord$, define $I_e(P)$, $\forall e \in E$, as
\begin{align} \label{Ie}
	& \underset{i \in S_e^A(P)}{\sum} \chi_{e}(i, \{j:\ord(j)\leq\ord(i), j\in S_e^A(P)\}) -   \nonumber \\
& \hspace{3cm} \underset{i \in S_e^{A\backslash B}(P)}{\sum} \chi_{e}(i, \{j:\ord(j)\leq\ord(i), j\in S_e^{A\backslash B}(P)\}).
\end{align}
Hart and Mas-Collel~\cite{10.2307/1911054} proved that the potential is independent of the ordering $\ord$ that players are considered. As mentioned before,  $\Phi^A(P)$ is a restriction of $\Phi(P)$ where only players in $A$ participate. Thus,  independence from $\ord$ also applies to the limited potential.

Firstly, we focus on the first term of \eqref{Ie} and choose an ordering where the players in set $A$ are first. Then we observe that  by substituting $S^A_e(P)$ with  $S_e(P)$, the cost share  remains the same. This is due to the fact that any player coming after the players in set $A$ in the ordering has no impact in the cost  computation. These are the players who belong in set $ N \setminus A$ (since we assume players in $A$ are first). Therefore, the first term of \eqref{Ie} equals to
\begin{align*}
	\underset{i \in S_e^A(P)}{\sum} \chi_{e}(i, \{j:\ord(j)\leq\ord(i), j\in S_e(P)\}) .
\end{align*}
Following the same technique for the second term of \eqref{Ie}, we choose an ordering in which the  players in $A \setminus B$  are first. Then we can substitute $S_e^{A \setminus B}(P)$ with $S_e^{N \setminus B}(P)$ without affecting the term's value. Therefore,  \eqref{Ie} is equivalent to
\begin{align} \label{IIe}
	&\underset{i \in S_e^A(P)}{\sum} \chi_{e}(i, \{j:\ord(j)\leq\ord(i), j\in S_e(P)\}) -   \nonumber \\
	& \hspace{3cm} \underset{i \in S_e^{A\backslash B}(P)}{\sum} \chi_{e}(i, \{j:\ord(j)\leq\ord(i), j\in S_e^{N\backslash B}(P)\}).
\end{align}
For each $e\in E$, define $I_e'(P)$ to be equal to 
\begin{align} \label{Ie'}
	&\underset{i \in S_e^{N\backslash A}(P)}{\sum} \Big( \chi_{e}(i, \{j:\ord(j)\leq\ord(i), j\in S_e(P)\}) - \nonumber \\
& \hspace{4cm}\chi_{e}(i, \{j:\ord(j)\leq\ord(i), j\in S_e^{N\backslash B}(P)\}) \Big).
\end{align}
Note that $I'_e(P) \geq 0$, $\forall e\in E$. Intuitively, the first term computes the cost with respect to all players using resource $e$, $S_e(P)$. Regarding the second term, if we take away some of these players, i.e., players in $B$, then due to convexity the costs of the remaining players either  remain the same or are reduced. This depends on the position  players in $B$ had in the ordering. To simplify, for the rest of this proof, let   
\begin{align}
	&\chi^{N}_i(P) = \chi_{e}(i, \{j:\ord(j)\leq\ord(i), j\in S_e(P)\}) \label{subn},\\
	&\chi^{N \setminus B}_i(P) = \chi_{e}(i, \{j:\ord(j)\leq\ord(i), j\in S_e^{N\backslash B}(P)\}).\label{subnb}
\end{align}
Since $I_e'(P) \geq 0$, we get that for each $ e\in E$, 
\begin{align*}
I_e(P) \leq I_e(P) + I_e'(P)
\end{align*}
 which, by \eqref{IIe}, \eqref{Ie'}, \eqref{subn} and \eqref{subnb}, is equivalent to 
\begin{align}
	&\underset{i \in S_e^A(P)}{\sum}  \chi^{N}_i(P) -     \underset{i \in S_e^{A\backslash B}(P)}{\sum}\chi^{N \setminus B}_i(P) \leq    \nonumber    \\
	&  \quad  \leq \underset{i \in S_e^A(P)}{\sum} \chi^{N}_i(P) -     \underset{i \in S_e^{A\backslash B}(P)}{\sum}\chi^{N \setminus B}_i(P) +   \underset{i \in S_e^{N\backslash A}(P)}{\sum} \left( \chi^{N}_i(P) - \chi^{N \setminus B}_i(P) \right).\label{p1}
\end{align}
By the assumption $B \subseteq A \subseteq N$, we get that $(N \setminus A) \cup (A \setminus B) = N \setminus B$. Thus inequality \eqref{p1} becomes
\begin{align*} 
	& \underset{i \in S_e^A(P)}{\sum} \chi^{N}_i(P) -     \underset{i \in S_e^{A\backslash B}(P)}{\sum}\chi^{N \setminus B}_i(P) 
	\leq  \underset{i \in S_e(P)}{\sum} \chi^{N}_i(P)-  \underset{i \in S_e^{N\backslash B}(P)}{\sum} \chi^{N \setminus B}_i(P).
\end{align*}
Substituting $\chi^{N}_i(P)$ and  $\chi^{N \setminus B}_i(P)$ from \eqref{subn} and \eqref{subnb}, we get by \eqref{IIe} that the previous is equivalent to
\begin{align*}
	I_e(P) \leq   \Phi_e (P) - \Phi_e^{N \setminus B}(P)  \quad 
	\Leftrightarrow \quad  \sum_{e\in E} I_e(P) \leq   \sum_{e\in E} \Phi_e (P) - \Phi_e^{N \setminus B}(P).
\end{align*}
By \eqref{cl}, we conclude to the desirable  ~$\Phi_B^A(P)\leq \Phi_B(P)$.

\item
By definition \eqref{ppot} of partial potential, we have
\begin{align} \label{pp}
	& \Phi_B^A(P) =  \Phi^A(P) - \Phi^{A\backslash B}(P) =  \underset{e\in E}{\sum} \left(\Phi^A_e(P) - \Phi^{A\backslash B}_e(P) \right).
\end{align}
For each $e \in E$ and any $A' \subseteq A$, observe that $S^{A'}_e(P)= S^{A'}_e(P')$. Thus
\begin{align*}
	&\underset{i \in S_e^A(P)}{\sum} \chi_{e}(i, \{j:\ord(j)\leq\ord(i), j\in S_e^A(P)\}) \\
&=\underset{i \in S_e^A(P')}{\sum} \chi_{e}(i, \{j:\ord(j)\leq\ord(i), j\in S_e^A(P')\}).
	\end{align*}
Similarly,  we prove that $\Phi^{A\backslash B}_e(P) =  \Phi^{A\backslash B}_e(P')$. Therefore, using \eqref{pp}, we have
\begin{align*} 
	\Phi_B^A(P) =   \underset{e\in E}{\sum} \left(\Phi^A_e(P') - \Phi^{A\backslash B}_e(P') \right)=\Phi_B^A(P'). 
\end{align*}

\item 
Let $P$ be an outcome of the game. Her contribution in the potential value is given by
\begin{align} \label{last}
	\Phi_z(P) = \Phi(P) - \Phi^{N\setminus \{z\}}(P) =  \underset{e\in E}{\sum} \left(\Phi_e(P) - \Phi^{N\backslash \{z\}}_e(P) \right) = \underset{e\in E}{\sum} I_e(P),
\end{align}
where $I_e(P)$ equals 
\begin{align*}
	&\underset{i \in S_e(P)}{\sum} \chi_{e}(i, \{j:\ord(j)\leq\ord(i), j\in S_e(P)\}) \\
	& \hspace{4cm} - \underset{i \in S_e^{N\setminus \{z\}}(P)}{\sum} \chi_{e}( i, \{j:\ord(j)\leq\ord(i), j\in S_e^{N\setminus \{z\}} \}). 
\end{align*}
Since the potential is independent of the players ordering, we choose the $\ord$ such that player $z$ is last. Then \eqref{last} equals to
\begin{align*}
	\underset{e\in E}{\sum} \chi_{e}(z, \{j:\ord(j)\leq\ord(z), j\in S_e(P)\}) 
	&= \underset{e\in E}{\sum} \chi_{e}(z,j:j\in S_e(P))  \\
&= \underset{e\in E}{\sum} \chi_{ze}(P)= X_{z}(P). 		\end{align*}
which completes the proof.\hfill 

\end{enumerate}
\hfill 
\end{proof}
\vspace{3mm}

\begin{proof}[Proof of Proposition~\ref{prop:diffpotentialcosts}]
	By definition of the partial potential \eqref{ppot},
	\begin{align*}
		\Phi_B(P) - \Phi_B(P') = \Phi(P) - \Phi^{N\setminus B}(P) - \left( \Phi(P') - \Phi^{N\setminus B}(P') \right)= \Phi(P) - \Phi(P').
	\end{align*}
	Since the underlying game (considering all players in $N$) is a potential game \cite{DBLP:journals/teco/KolliasR15},  ~$\Phi(P) - \Phi(P')= X_i(P) - X_i(P')$.
	\hfill 
\end{proof}
\vspace{3mm}

\begin{proof}[Proof of Lemma~\ref{lemma:resource_potential_sum_costs}]
	By definition  \eqref{ppot}, we have 
	\begin{align}
\Phi_{e,B}(P)  = \Phi_e(P) - \Phi_e^{N\setminus B}(P)
 = \underset{e\in E}{\sum} \left(\Phi_e(P) - \Phi^{N\backslash B}_e(P) \right) = I_e(P). \label{d}
\end{align}
where $I_e(P)$ equals to
\begin{align} \label{Ie2}
&\underset{i \in S_e(P)}{\sum} \chi_e(i, \{j:\ord(j)\leq\ord(i), j\in S_e(P)\}) \nonumber\\
 &\hspace{3cm} - \underset{i \in S_e^{N\setminus B}(P)}{\sum} \chi_{e}( i, \{j:\ord(j)\leq\ord(i), j\in S_e^{N\setminus B} \}). 
\end{align}
	Then we  break the first term of \eqref{Ie2} to the  sum of 
	\begin{align*}
		&\underset{i \in S_e^{N\backslash B}(P)}{\sum} \chi_{e}(i, \{j:\ord(j)\leq\ord(i), j\in S_e(P)\}) \\
&\hspace{5cm}+   \underset{i \in S_e^{B}(P)}{\sum} \chi_{e}(i, \{j:\ord(j)\leq\ord(i), j\in S_e(P)\}).
	\end{align*}
	We choose an ordering $\ord$ in which all players in $N \setminus B$ come first. Then the previous sum is equivalent to
	\begin{align*}
		&\underset{i \in S_e^{N\backslash B}(P)}{\sum} \chi_{e}(i, \{j:\ord(j)\leq\ord(i), j\in S^{N \setminus B}_e(P)\}) \\
& \hspace{5cm}+   \underset{i \in S_e^{B}(P)}{\sum} \chi_{e}(i, \{j:\ord(j)\leq\ord(i), j\in S_e(P)\}). 
	\end{align*}
	
	Substituting the previous to the first term of \eqref{Ie2}  gives
	\begin{align*}
		\underset{i \in S_e^{B}(P)}{\sum} \chi_{e}(i, \{j:\ord(j)\leq\ord(i), j\in S_e(P)\}).
	\end{align*}
	Combining it with the definition of $I_e(P)$ yields to
	\begin{align*}
		I_e (P)&=  \underset{i \in S_e^{A}(P)}{\sum} \chi_{e}(i, \{j:\ord(j)\leq\ord(i), j\in S_e(P)\}) \\
		&\leq \underset{i \in S_e^{A}(P)}{\sum} \chi_{e}(i, j: j \in S_e(P)) = \underset{i \in S_e^{A}(P)}{\sum}  \chi_{ie}(P) =  \underset{i\in A}{\sum} \chi_{ie}(P).
	\end{align*}
	Equation \eqref{d} completes the proof of the lower bound.
	
	For the upper bound consider a fixed ordering of the players in $B$. The partial potential can be written as
\begin{align}
	\Phi_{e,B}(P)
	&= \left( \Phi_e(P)-\Phi_e^{N\setminus B}(P)\right) \notag \\
	&= \sum_{i \in S_e^B(P)} 
	\chi_{e}\left(i, \left\{j:\ord(j)\leq\ord(i); j\in S_e^B(P)\right\}\cup S_e^{N\setminus B}(P) \right) \notag\\
	&\geq \int_{\weightsSumSet{\fixedBPlayers}{\state}}^{\weightsSumSet{\allPlayers}{\state}} c_e(x) dx \notag\\
	&\geq \left[\frac{x\cdot c_e(x)}{d+1}\right]_{\weightsSumSet{\fixedBPlayers}{\state}}^{\weightsSumSet{\allPlayers}{\state}}\notag\\
	&= \frac{\weightsSumSet{\allPlayers}{\state}\cdot c_e(\weightsSumSet{\allPlayers}{\state}) 
		- \weightsSumSet{\fixedBPlayers}{\state}\cdot c_e(\weightsSumSet{\fixedBPlayers}{\state}) }
	{d+1} \notag\\
	&= \frac{\weightsSum{\state}\cdot c_e(\weightsSum{\state}) }{d+1} - 		
	\frac{ \weightsSumSet{\fixedBPlayers}{\state}\cdot c_e(\weightsSumSet{\fixedBPlayers}{\state}) }
	{d+1} \notag\\
	&= \frac{\sum_{i \in N} \chi_{ie}(P)}{d+1} - 		
	\frac{ \weightsSumSet{\fixedBPlayers}{\state}\cdot c_e(\weightsSumSet{\fixedBPlayers}{\state}) }
	{d+1}, \label{eq_lowerpot}
\end{align}
where the first inequality follows by repeatedly applying Proposition \ref{prop:properties}(c) and \ref{prop:properties}(d) and adding additional players of weight 0 (which do not change the cost shares). The second inequality holds, 
since $c_e$ is a polynomial of maximum 
degree $d$ with non-negative coefficients.
	
Observe, that $\weightsSumSet{\fixedBPlayers}{\state}\cdot c_e(\weightsSumSet{\fixedBPlayers}{\state})$ is the social cost of $P$ on resource $e$ if only the players 
in $N\setminus B$ are in the game. By Proposition \ref{prop:properties}(a), 
the cost shares of those players can only increase if the players in $B$ are 
joining the game, i.e.:
\begin{align*}
	\weightsSumSet{\fixedBPlayers}{\state}\cdot c_e(\weightsSumSet{\fixedBPlayers}{\state})
	\leq \sum_{i \in N\setminus A} \chi_{ie}(P). 
\end{align*}
Combining this with \eqref{eq_lowerpot} completes the proof of the claim:
\begin{align*}
	\Phi_{e,B}(P) \geq \frac{\sum_{i \in N} \chi_{ie}(P)}{d+1} - \frac{ \sum_{i \in N\setminus B} \chi_{ie}(P)}{d+1} = \frac{\sum_{i \in B} \chi_{ie}(P)}{d+1}
\end{align*}
\hfill 
\end{proof}
\vspace{3mm}

\begin{proof}[Proof of Corollary~\ref{cor:potential_sum_costs}]
By the definition of the partial potential \eqref{ppot} and by applying Lemma~\ref{lemma:resource_potential_sum_costs}, we directly have
\begin{align*}
	\Phi_B(P) =  \sum_{e \in E} \Phi_{e,B}(P) \leq  \sum_{e \in E} \sum_{i \in B} \chi_{ie}(P) = \sum_{i \in B} X_i(P)
\end{align*}

and 
\begin{align*}
	\sum_{i\in B} X_i(P) =\sum_{i \in B} \sum_{e \in E}  \chi_{ie}(P) =\sum_{e \in E} \sum_{i \in B}   \chi_{ie}(P) & \leq \sum_{e \in E} \Phi_{e,B}(P) \cdot (d+1) \\
	&= \Phi_B(P) \cdot (d+1).
\end{align*}
\hfill 
\end{proof}
\vspace{3mm}

\section{Proofs for the Approximation, PoA and Stretch in Section~\ref{section:Approximation}~and~\ref{section:PoA_Stretch}}

\begin{proof}[Proof of Lemma~\ref{lemma:shatoprop}]

Since $c_e$ is a polynomial of maximum degree $d$ with non-negative coefficients, 
it suffices  to show the inequalities for all monomial cost functions $c_e(x)=x^r$, with $r=\{0, \ldots, d\}$. 
Fix some resource $e$ with monomial cost function and a player $i$ assigned to $e$, i.e., $e \in P_i$. 
Denote $Y=\{j \neq i: e \in P_j\}$ and $w=w_i$. Define $y=\sum_{j\in Y}w_j$ and $z=\frac{w}{y}$. 
By Proposition \ref{prop:properties} $(b)$, we can upper bound $\chi_{ie}(P)$ by replacing $Y$ 
with a single player of weight $y$, i.e.,
\begin{align}
\chi_{ie}(P) &\leq \frac{1}{2} \left( (y+w)^{r+1} - y^{r+1} \right) + \frac{1}{2}\cdot w^{r+1}
 = y^{r+1} \cdot \frac{1}{2}\cdot \left( (z+1)^{r+1} - 1 + z^{r+1} \right)  \nonumber\\
&= y^{r+1} \cdot   \left(z^{r+1} +  \frac{1}{2}\cdot \sum_{j=1}^r\binom{r+1}{j}\cdot  z^j \right) \nonumber  =:A.
\end{align}
Similarly, by repeatedly using Proposition \ref{prop:properties} $(c)$ and by adding additional players of weight 0, we can lower bound $\chi_{ie}(P)$ by 
\begin{align*} 
&\frac{1}{y}\cdot \int_{0}^{y}\left((x+w)^{r+1} - x^{r+1}\right)dx 
= \frac{1}{y} \cdot \frac{1}{r+2} \cdot \left((y+w)^{r+2} - y^{r+2} - w^{r+2}\right)  \nonumber\\
&= y^{r+1} \cdot \frac{1}{r+2}  \cdot \left(\left(z+1\right)^{r+2} - 1 - z^{r+2}\right) \nonumber 
=  y^{r+1} \cdot \frac{1}{r+2}  \cdot \sum_{j=1}^{r+1} \binom{r+2}{j} \cdot z^j =:B. 
\end{align*}
The proportional cost of player $i$,  $\chi_{ie}^{\text{Prop}}(P)$, equals to
\begin{align*} 
  w\cdot c_e(y+w) = w \cdot(y+w)^r = y^{r+1} \cdot z \cdot ( z+1)^r = y^{r+1}  \cdot \sum_{j=1}^{r+1} \binom{r}{j-1} \cdot z^{j}.
\end{align*}
To complete the proof we give an upper bound on $\frac{A}{\chi_{ie}^{\text{Prop}}(P)}$ and a lower bound on $\frac{B}{\chi_{ie}^{\text{Prop}}(P)}$. We have,
\begin{align*}
\frac{A}{ \chi_{ie}^{\text{Prop}}(P)} 
& =\frac{z^{r+1} +  \frac{1}{2} \sum_{j=1}^{r}\binom{r+1}{j}\cdot  z^j}
             { \sum_{j=1}^{r+1} \binom{r}{j-1} \cdot z^{j}}
 =\frac{z^{r+1} +  \frac{1}{2} \sum_{j=1}^{r}\binom{r+1}{j}\cdot  z^j}
             { z^{r+1} + \sum_{j=1}^{r} \binom{r}{j-1} \cdot z^{j}},
\end{align*}
which is upper bounded by 
\begin{align}
\frac{A}{ \chi_{ie}^{\text{Prop}}(P)} \leq \max\left(1, \max_{1\leq j\leq r} \ \frac{ \binom{r+1}{j}}{2 \cdot \binom{r}{j-1}}\right) = \max\left(1, \max_{1\leq j\leq r} \ \frac{r+1}{2 \cdot j}\right)  \leq \frac{d+1}{2}.\label{ubSoverP}
\end{align}
This implies the lower bound on $\chi_{ie}^{\text{Prop}}(P)$ in the statement of the lemma.
On the other hand, by first order conditions, 
\begin{align*}
\frac{B}{ \chi_{ie}^{\text{Prop}}(P)} 
& =\frac{\frac{1}{r+2}  \cdot \sum_{j=1}^{r+1} \binom{r+2}{j} \cdot z^j}
        { \sum_{j=1}^{r+1} \binom{r}{j-1} \cdot z^{j}},
\end{align*}
which achieves its extreme values at the roots of
\begin{align*}
	g(z) := \sum_{j=1}^{r+1} \sum_{k=1}^{r+1} (j-k) \binom{r+2}{j} \binom{r}{k-1} \cdot z^{k+j-1}.
\end{align*}

\begin{claim}\label{uniqueRoot}
The function $g: z \rightarrow \sum_{j=1}^{r+1} \sum_{k=1}^{r+1} (j-k) \binom{r+2}{j} \binom{r}{k-1} \cdot z^{k+j-1}$  has a unique positive real root at $z=1$.
\end{claim}
\begin{proof}
	We will show that g(z) has a unique positive real root at $z=1$, is negative for $z<1$ and positive for $z>1$.
	To this end, by combining coefficients of the same monomial, we get
	\begin{align*}
	g(z) 
	& = \sum_{\sigma=2}^{r+1} \sum_{j=1}^{\sigma-1} (2j-\sigma) \binom{r+2}{j} \binom{r}{\sigma-j-1} \cdot z^{\sigma-1}\\
	& \phantom{=}  + \sum_{\sigma=r+3}^{2r+2} \sum_{j=\sigma-r-1}^{r+1} (2j-\sigma) \binom{r+2}{j} \binom{r}{\sigma-j-1} \cdot z^{\sigma-1}, 
	\end{align*}
	where by symmetry the coefficient for $\sigma=r+2$ is $0$.
	Pairing summands $j$ and $\sigma-j$, we get
	\begin{align*}
	&g(z) =  \sum_{\sigma=2}^{r+1} \sum_{j=1}^{\lfloor\frac{\sigma-1}{2}\rfloor} 
	(2j-\sigma)
	\left(
	\binom{r+2}{j} \binom{r}{\sigma-j-1}
	- \binom{r+2}{\sigma-j} \binom{r}{j-1}
	\right)   \cdot z^{\sigma-1}\\
	& \phantom{=}  + \sum_{\sigma=r+3}^{2r+2} \sum_{j=\lceil\frac{\sigma}{2}\rceil}^{r+1} 
	(2j-\sigma)
	\left(
	\binom{r+2}{j} \binom{r}{\sigma-j-1}
	-  \binom{r+2}{\sigma-j} \binom{r}{j-1}
	\right)   \cdot z^{\sigma-1}.
	\end{align*}
	Define $\beta(\sigma,j):=(2j-\sigma) \cdot
	\left(
	\binom{r+2}{j} \binom{r}{\sigma-j-1}
	- \binom{r+2}{\sigma-j} \binom{r}{j-1}
	\right)$.
	Now observe that 
	\begin{align*}
	\binom{r+2}{j} \binom{r}{\sigma-j-1} = \frac{(\sigma-j)(r+2-(\sigma-j))}{j(r+2-j)}\cdot \binom{r+2}{\sigma-j} \binom{r}{j-1}.
	\end{align*}
	Since $\frac{(\sigma-j)(r+2-(\sigma-j))}{j(r+2-j)}\geq 1$ for all $(\sigma,j)$ where $2\leq\sigma\leq r+1$ and $1\leq j \leq \frac{\sigma-1}{2}$ and for all $(\sigma,j)$ where $r+3 \leq \sigma \leq 2r+2$ and $\frac{\sigma}{2}\leq j \leq r+1$, we get that 
	$\beta(\sigma,j)\leq 0$ when $\sigma \leq r+1$ and $\beta(\sigma,j)\geq 0$ when $\sigma \geq r+3$ for all $j$ in the corresponding range.  \emph{Descartes' rule of signs} implies that $g(z)$ has at most one positive real root. Simple arithmetic shows that $z=1$ is a root of $g(z)$. 
	\hfill 
\end{proof}
\vspace{3mm}

By the previous lemma, we conclude that 
$\frac{B}{ \chi_{ie}^{\text{Prop}}(P)}$ is minimized for $z=1$, i.e.,
\begin{align*}
	\frac{B}{ \chi_{ie}^{\text{Prop}}(P)} 
        \geq \frac{\frac{1}{r+2}  \cdot \sum_{j=1}^{r+1} \binom{r+2}{j} }
        { \sum_{j=1}^{r+1} \binom{r}{j-1}}
        = \frac{\frac{1}{r+2} \cdot (2^{r+2}-2)}
               {2^r}
         \geq \frac{4}{r+3}
         \geq \frac{4}{d+3},
\end{align*}
which completes the proof of the upper bound in the lemma. \hfill 
\end{proof}
\vspace{3mm}

\begin{proof}[Proof of Lemma~\ref{lemma:approximate_equilibria}]
	Let $P$ be a $\rho$-approximate equilibrium in the SV weighted congestion game. Using the equilibrium condition and Corollary~\ref{corollary:shatoprop}, we have
	\begin{align*}
	X_i^{\text{Prop}}(P) 
	\leq\frac{d+3}{4} \cdot X_i(P)
	\leq\frac{d+3}{4}\cdot  \rho \cdot X_i(P) 
	\leq\frac{d+3}{4} \cdot \frac{d+1}{2}\cdot  \rho \cdot X_i^{\text{Prop}}(P)
	\end{align*}
\end{proof}
\vspace{3mm}

\begin{proof}[Proof of Lemma~\ref{lemma:approx-PoA}]
	Let $P$ be an $\rho$-approximate pure Nash equilibrium and $P^*$ the optimal outcome:
	\begin{align}
	SC(P) =  \sum_{i \in N} \sum_{e \in P_i} \chi_e(i, S_e(P))
	\overset{\text{Def. } \rho \text{-PNE}}{\le}
	\rho \cdot  \sum_{i \in N} \sum_{e \in P^*_i} \chi_e(i, S_e(P) \cup \{i\}).\nonumber
	\end{align}
	Due to the convexity of the cost functions, note that the cost share of any player on any resource is always upperbounded by the marginal cost increase she causes to the resource cost when she is last in the ordering, $\chi_e(i, S_e(P) \cup \{i\}) \leq  C_e(f_e(P)+ w_i) - C_e(f_e(P))$. Thus,
	\begin{align}
	SC(P)
	& \le \rho \cdot \left( \sum_{i \in N} \sum_{e \in P^*_i} C_e(f_e(P)+ w_i) - C_e(f_e(P)) \right) \nonumber\\
	& \le \rho \cdot \left( \sum_{e\in E}\sum_{i: e\in P^*_i} C_e(f_e(P)+ w_i) - C_e(f_e(P)) \right) \nonumber\\
	& \le \rho \cdot \left( \sum_{e\in E} C_e(f_e(P)+ f_e(P^*)) - C_e(f_e(P)) \right).\label{smooth}
	\end{align}
	The last inequality follows from   assumption that $C_e$ is a convex function  in players' weights. 
	
	\begin{claim}Let $\lambda = 2^{\frac{d}{d+1}} \cdot \left(2^{\frac{1}{d+1}}-1\right)^{-d}$ and $\mu=2^{\frac{d}{d+1}}-1$, then  for $x,y > 0$ and $d\geq 1$,
		$(x+y)^{d+1} -x^{d+1} \leq \lambda \cdot y^{d+1} + \mu \cdot x^{d+1}$.
	\end{claim}
	Using this claim that was proven in~\cite{DBLP:conf/wine/GkatzelisKR14}, \eqref{smooth} becomes
	\begin{align*}
	SC(P) 
	&\leq \rho\cdot \left( \sum_{e\in E} \lambda \cdot C_e(f_e(P^*)) + \mu \cdot C_e(f_e(P)) \right) \\
	&=  \rho \cdot \lambda \cdot SC(P^*) +  \rho \cdot \mu \cdot SC(P). 
	\end{align*}
	Rearranging  and substituting the values for $\lambda$ and $\mu$ we get an upper bound on the $\rho$-PoA,
	\begin{align*}
	\rho\text{-PoA}
	&\le \frac{\rho \cdot \lambda}{1-\rho \cdot \mu}
	= \frac{\rho \cdot 2^{\frac{d}{d+1}} \cdot \left(2^{\frac{1}{d+1}}-1\right)^{-d}}{1-\rho \cdot \left( 2^{\frac{d}{d+1}}-1\right)}
	= \rho \cdot \frac{2}{2^{\frac{1}{d+1}}} \cdot \frac{\left( 2^{\frac{1}{d+1}} - 1 \right)^{-d}}{1 - \rho \cdot \frac{2}{2^{\frac{1}{d+1}}} + \rho}\\
	& =  \frac{2 \cdot \rho \left( 2^{\frac{1}{d+1}} - 1\right)^{-d}}{2^{\frac{1}{d+1}} \cdot (1 + \rho) - 2 \cdot \rho}
	= \frac{\rho \cdot (2^{\frac{1}{d+1}} - 1)^{-d}}{2^{-\frac{d}{d+1}} \cdot( 1+ \rho) - \rho}.
	\end{align*}
	\hfill 
\end{proof}
\vspace{3mm}

\begin{proof}[Proof of Lemma~\ref{lemma:approx-stretch}]
Let $P$ be a $\rho$-approximate equilibrium, $P^*$ the optimal outcome and $\hat{P}=\min_{P'\in\CP}\Phi(P')$ the minimizer of the potential which is by definition a pure Nash equilibrium. Then the $\rho$-approximate price of anarchy equals to
\begin{align*}
	\rho\text{-PoA} = \underset{P \in \rho\text{-PNE}}{\max} \frac{SC(P)}{SC(P^*)}
	\ge \underset{P \in \rho\text{-PNE}}{\max} \frac{SC(P)}{SC(\hat{P})}
	\overset{\text{Def.}~\Phi}{\ge} \underset{P \in \rho\text{-PNE}}{\max} \frac{\Phi(P)}{SC(\hat{P})}.
\end{align*}
By  Lemma~\ref{lemma:approx-PoA} and Corollary~\ref{cor:potential_sum_costs} for $A=N$, the $\rho$-PoA is bounded as follows
\begin{align*}
	\underset{P \in \rho\text{-PNE}}{\max} \frac{\Phi(P)}{(d+1)\cdot \Phi(\hat{P})}
	\leq \rho\text{-PoA}
	\leq \frac{\rho \cdot (2^{\frac{1}{d+1}} - 1)^{-d}}{2^{\frac{-d}{d+1}} \cdot( 1+ \rho) - \rho}.
\end{align*}
Rearranging the terms gives the desired upper bound of the $\rho$-stretch,
\begin{align*}
	\rho\text{-}\Omega=\underset{P \in \rho\text{-PNE}}{\max} \frac{\Phi(P)}{\Phi(\hat{P})}
	\leq \frac{\rho \cdot (2^{\frac{1}{d+1}} - 1)^{-d}\cdot (d+1)}{2^{-\frac{d}{d+1}} \cdot( 1+ \rho) - \rho}.
\end{align*}
\hfill 
\end{proof}
\vspace{3mm}

\begin{proof}[Proof of Lemma~\ref{lemma:approx-limited-stretch-social-costs}]
	To show the lemma we lower and upper bound the $D$-partial potential.
	Let $e$ be an arbitrary resource. By using Lemma~\ref{lemma:resource_potential_sum_costs} and Lemma~\ref{lemma:shatoprop}, we get
	\begin{align}
	\Phi_{e, D} (P)
	\leq \sum_{i \in D} \chi_{ie}(P)
	\leq \frac{d+1}{2} \cdot  \sum_{i \in D} \chi_{ie}^{\text{Prop}}(P). \label{0}
	\end{align}
By definition of the proportional share $\chi_{ie}^{\text{Prop}}$, \eqref{0} becomes
	\begin{align}
	\Phi_{e, D} (P)
	\leq & ~\frac{d+1}{2} \cdot \sum_{i \in D} w_i \cdot c_e(f_e(P))
	= \frac{d+1}{2} \cdot  f_e^D(P)  \cdot c_e(f_e(P)) \nonumber\\
	&= \frac{d+1}{2} \cdot \frac{f_e^D(P)}{f_e(P)} \cdot f_e(P) \cdot c_e(f_e(P))
	= \frac{d+1}{2} \cdot \frac{f_e^D(P)}{f_e(P)} \cdot \sum_{i\in N}\chi_{ie}(P). \label{00}
	\end{align}
Rearranging \eqref{00} gives a relation of the per unit contribution to $\Phi_D$ and $\Phi$,
	\begin{align*}
	\frac{\Phi_{e, D} (P)}{f_e^D(P)}
	\leq \frac{d+1}{2}\cdot \frac{\sum_{i\in N}\chi_{ie}(P)}{f_e(P)},
	\end{align*}
	and by summing up over all resources $e$, we get
	\begin{align}
	\frac{\Phi_{D} (P)}{W_D}
	\leq \frac{d+1}{2}\cdot \frac{SC(P)}{W},\label{eq:stretch-upper}
	\end{align}
	where $W= \sum_{i \in N} w_i = \sum_{e \in E} f_e(P)$ and $W_D = \sum_{i \in D} w_i = \sum_{e \in E} f_e^D(P)$.

	Similar to \eqref{00}, we lower bound the $D$-partial potential with
	\begin{align*}
	\Phi_{e, D} (P)
	\geq \frac{1}{d+1} \cdot  \sum_{i \in D} \chi_{ie}(P)
	\geq& ~\frac{4}{(d+1) \cdot (d+3)} \cdot \sum_{i \in D} w_i \cdot c_e(f_e(P))  \\
	&= \frac{4}{(d+1)\cdot (d+3)} \cdot \frac{f_e^D(P)}{f_e(P)} \cdot \sum_{i\in N}\chi_{ie}(P).
	\end{align*}
The first inequality uses Lemma~\ref{lemma:resource_potential_sum_costs} and the second uses Lemma~\ref{lemma:shatoprop}.
	Again we get a per unit contribution to $\Phi_D$ and $\Phi$ on one resource and in the whole game,
	\begin{align}
\frac{\Phi_{e, D} (P)}{f_e^D(P)} &\geq  \frac{4}{(d+1)\cdot(d+3)} \cdot \frac{\sum_{i\in N}\chi_{ie}(P)}{f_e(P)} \notag \\
	 \Leftrightarrow ~\frac{\Phi_{D} (P)}{W_D} &\geq   \frac{4}{(d+1)\cdot (d+3)} \cdot  \frac{SC(P)}{W}.\label{eq:stretch-lower}
	\end{align}
Combining \eqref{eq:stretch-upper} with \eqref{eq:stretch-lower} and rearranging  the terms completes Lemma's \ref{lemma:approx-limited-stretch-social-costs} proof,
	\begin{align*}
	\frac{\Phi_D(P)}{\Phi_D(\hat{P})}
	&\leq \frac{d+1}{2} \cdot \frac{SC(P)}{W} \cdot \frac{W_D}{1}  \cdot \frac{(d+1) \cdot(d+3)}{4}  \cdot \frac{W}{SC(\hat{P})}\cdot  \frac{1}{W_D} \notag \\
	& =  \frac{(d+1)^2 \cdot (d+3)}{8} \cdot  \frac{SC(P)}{SC(\hat{P})}.  
	\end{align*}
	
	\vspace{-8mm}
	\hfill 
\end{proof}

\section{Proofs for the Computation in Section~\ref{section:computation}}

\begin{proof}[Proof of Lemma~\ref{lemma:potential_costs}]
Let $D_r^i \subseteq D_r$ the set of players who still have to perform their last move after player $i$ in phase $r$. Then by definition of the partial potential \ref{pot}, $\potentialDPhase{\statePhase}$ equals to
\begin{align}\label{part1}
	\potentialAll{\statePhase} - \potentialFixed{\statePhase}
	= \sum_{i=1}^{\left|\movingPlayersPhase\right|}\left(\Phi^{N \backslash D_r^i}(\statePhase) - \Phi^{N \backslash D_r^{i-1}}(\statePhase)\right)
	= \sum_{i=1}^{\left|\movingPlayersPhase\right|} \Phi_i^{N \backslash D_r^i}(\statePhase).
\end{align}
For each player $i$, her strategy in state $\statePhase$ is identical to her strategy in $\state^{r,i}$. By Proposition~\ref{prop:potentials}~(\ref{prop:partialtofull}), ~\ref{prop:potentials}~(\ref{prop:sameStrategies}) and~\ref{prop:potentials}~(\ref{prop:potentialcosts}), we  upperbound \eqref{part1} by
\begin{align*}
	\sum_{i=1}^{\left|\movingPlayersPhase\right|} \Phi_i^{N \backslash D_r^i}(\statePhase)
	= \sum_{i=1}^{\left|\movingPlayersPhase\right|} \Phi_i^{N \backslash D_r^i}(\state^{r,i})
	\leq \sum_{i=1}^{\left|\movingPlayersPhase\right|} \Phi_i(\state^{r,i})
	= \sum_{i=1}^{\left|\movingPlayersPhase\right|} \playerCosts{\state^{r,i}}.
\end{align*}
\hfill 
\end{proof}
\vspace{3mm}

\begin{proof}[Proof of Lemma~\ref{lemma:potential_block}]
We show the lemma by contradiction. Thus, assume that $\potentialDPhase{\state^{r-1}} > \frac{n}{\constant} \cdot b_r$.
Let $S_r$, $T_r \subseteq \movingPlayersPhase$, be the set of players whose last move is an \sMove \ and a \tMove, accordingly, such that $S_r \cup T_r = \movingPlayersPhase$.
First, we focus on the players in $S_r$. Let $i \in S_r$ be an arbitrary player. By  definition of an \sMove, player $i$ decreases her costs in her last move during phase $r$ by at least $(s-1) \cdot \playerCosts{\state^{r,i}}$. By Proposition~\ref{prop:diffpotentialcosts}, any such improvement step also decreases the $i$-partial potential by the same amount.  Summing up over all players $i \in S_r$, we get a lower bound on the total decrease of the $D_r$-partial potential between states $P^{r-1}$ and $P^r$: $\Phi_{\movingPlayersPhase}(P^{r-1}) - \Phi_{\movingPlayersPhase}(P^{r}) \geq (s-1) \cdot \sum_{i \in S_r} X_i(P^{r,i})$.
Rearranging, we upper bound the partial potential as follows,
\begin{align}
	\Phi_{\movingPlayersPhase}(P^{r})
	&\leq \Phi_{\movingPlayersPhase}(P^{r-1}) - (s-1) \cdot \sum_{i \in S_r} X_i(P^{r,i})\nonumber\\
	&\leq \Phi_{\movingPlayersPhase}(P^{r-1}) - (s-1) \cdot\left( \sum_{i \in D_r} X_i(P^{r,i}) - \sum_{i \in T_r} X_i(P^{r,i}) \right)\nonumber\\
	&\leq \Phi_{\movingPlayersPhase}(P^{r-1}) - (s-1) \cdot\left( \sum_{i \in D_r} X_i(P^{r,i}) - n \cdot b_r \right)\nonumber\\
	&\leq \Phi_{\movingPlayersPhase}(P^{r-1}) - (s-1) \cdot\left( \potentialDPhase{\statePhase} - n \cdot b_r \right)\nonumber\\
	&\leq \Phi_{\movingPlayersPhase}(P^{r-1}) - (s-1) \cdot\left(\potentialDPhase{\statePhase} - \constant \cdot \potentialDPhase{\state^{r-1}} \right)\nonumber\\		
	&\leq (1 + (s-1) \cdot \constant) \cdot \Phi_{\movingPlayersPhase}(P^{r-1}) - (s-1)\cdot \potentialDPhase{\statePhase}, \nonumber
\end{align}
where the third inequality follows from the fact that the cost of a player $i \in T_r$ is upper bounded by the block border $b_r$, the fourth inequality by Lemma~\ref{lemma:potential_costs} and the fifth one by the assumption. Rearranging the terms gives
\begin{align}
	\Phi_{\movingPlayersPhase}(P^{r}) \leq \frac{1+(s-1) \cdot \constant}{s} \cdot \Phi_{\movingPlayersPhase}(P^{r-1}).\label{eq:potential_phase}
\end{align}

Let $\bar{P}$ be an intermediate state between $P^{r-1}$ and $P^r$ such that all players in $S_r$ have already finished their $s$-move and play their strategies in $P^r$, while the  moving players in $T_r$ play their strategies in $P^{r-1}$. Consider a player $i \in T_r$. The difference in her cost after her $t$-move is at most $b_r$. This is due to the fact that her initial cost is at most $b_r$ (by the block construction) and the minimum cost she can improve to is zero. Then, by Proposition~\ref{prop:diffpotentialcosts}, the difference in the cost of player $i$ equals to the difference in the $i$-partial potential, that is, $ \Phi_i(\bar{P}) - \Phi_i(P^r)= X_i(P) - X_i(P') \leq b_r$. Summing up over all players in $T_r$, we get that the difference in the $D_r$-partial potential among states $\bar{P}$ and $P^r$ can be at most $n \cdot b_r$. Then, we get the following upper bound on the partial potential in state $\bar{P}$,
\begin{align} 
	\Phi_{\movingPlayersPhase}(\bar{P})
	&\leq \Phi_{\movingPlayersPhase}(P^r) + n \cdot b_r\nonumber
	\leq \frac{1+(s-1) \cdot \constant}{s} \cdot \Phi_{\movingPlayersPhase}(P^{r-1}) + \constant \cdot \potentialDPhase{\state^{r-1}}\nonumber\\
	&= \left(\frac{1 - \constant}{s} + 2  \cdot \constant \right)\cdot \Phi_{\movingPlayersPhase}(P^{r-1})\nonumber
	< \left(\frac{1 }{s} + 2 \cdot \constant \right)\cdot \Phi_{\movingPlayersPhase}(P^{r-1})\nonumber,
\end{align}
where the second inequality holds by \eqref{eq:potential_phase} and our assumption. Substituting $s$, we get
\begin{align*}
	\Phi_{\movingPlayersPhase}(\bar{P})
	< \frac{1}{\tStretch} \cdot \Phi_{\movingPlayersPhase}(P^{r-1}),
\end{align*}
which contradicts Corollary~\ref{stretch:upperbound}.
\hfill 
\end{proof}
\vspace{3mm}

\begin{proof}[Proof of Lemma~\ref{lemma:running_time}]
At the beginning of the algorithm's execution, the sum of all players' costs is at most $n \cdot X_{\max}$. By Corollary~\ref{cor:potential_sum_costs}, the potential is also upper bounded by the same amount. In the initial phase, each deviating player makes a $t$-move, therefore her cost improves by at least $(t-1) \cdot b_1$ (since her cost is at most $b_1$). The potential function also decreases by at least $(t-1) \cdot b_1$ in each step. Using the definition of $b_1$, we get that $(t-1)\cdot b_1 = \gamma \cdot g^{-1}\cdot X_{\max}$. Using both observations, we can compute the maximum number of improvement steps in the first phase,
\vspace{-.4cm}
\begin{align*}
	\frac{n \cdot X_{\max}}{\gamma\cdot g^{-1} \cdot X_{\max}}=n\cdot \gamma^{-1} \cdot g =n\cdot \gamma^{-1} \cdot \frac{2 \cdot n \cdot (d+1)}{\constant^3} = 2 \cdot n^2 \cdot (d+1) \cdot \gamma^{-4}.
\end{align*}
Consider an arbitrary phase $r \geq 1$. By Lemma~\ref{lemma:potential_block}, $\Phi_{\movingPlayersPhase}(P^{r-1}) \leq \frac{n}{\gamma} \cdot b_r$. Again, we look at the possible cost improvement in a deviation which equals to the potential decrease in this step. In this case, the cost improvement is at least $(t-1) \cdot b_{r+1}$. By definition of $b_{r+1}$, we have that $(t-1) \cdot b_{r+1} = b_r \cdot g^{-1} \cdot \gamma$. Similar,  the maximum number of improvement moves in this phase is
\vspace{-.4cm}
\begin{align*}
	\frac{ \frac{n}{\gamma} \cdot b_r}{b_r \cdot g^{-1}\cdot \gamma}=\frac{n \cdot g}{\gamma^{2}} =\frac{2 \cdot n^2 \cdot (d+1) \cdot \gamma^{-3}}{\gamma^{2}} = 2 \cdot n^2 \cdot (d+1) \cdot \gamma^{-5}.
\end{align*}
In total, we have at most ~$2 \cdot n^2 \cdot (d+1) \cdot \gamma^{-4} +\log\left(\frac{\playerCostsMax}{\playerCostsMin}\right)\cdot 2 \cdot n^2 \cdot (d+1) \cdot \gamma^{-5} = \left(1+\log\left(\frac{\playerCostsMax}{\playerCostsMin}\right)\right) \cdot 2 \cdot n^2 \cdot (d+1) \cdot \gamma^{-9}$~ improvement steps.
\hfill 
\end{proof}
\vspace{3mm}

\begin{proof}[Proof of Lemma~\ref{lemma:costs_end}]
We first show by contradiction the following. For $j \geq r$,  the increase in the cost of player $i$ from an arbitrary state $\state^j$  to state $\state^{j+1}$ is upper bounded by $\frac{n \cdot (d+1)}{\constant} \cdot b_{j+1}$. Thus, assume that $\playerCosts{\state^{j+1}} - \playerCosts{\state^j} > \frac{n \cdot (d+1)}{\constant} \cdot b_{j+1}$. Since  player $i$ does not deviate during phase $j+1$, the increase in her cost is caused by other players deviating to the resources she uses. Thus, there exists a set of resources $E' \subseteq E$ such that each resource in $E'$ is used by player $i$ and by at least one player in $D_{j+1}$ at state $\state^{j+1}$. This yields to
\begin{align*}
	& \sum_{e \in E'} \costshare{\state^{j+1}} > \frac{n \cdot (d+1)}{\constant} \cdot b_{j+1} \\
	&\Rightarrow
	\frac{\sum_{e \in E'}f_e(\state^{j+1}) \cdot c_e (f_e(\state^{j+1}))}{d+1}> \frac{n}{\constant} \cdot b_{j+1}\nonumber\\
	&\Leftrightarrow \frac{SC_{D_{j+1}}(\state^{j+1})}{d+1} > \frac{n}{\constant} \cdot b_{j+1}\\
	&\Rightarrow \Phi_{D_{j+1}}(\state^{j+1}) > \frac{n}{\constant} \cdot b_{j+1}.
\end{align*}
The last step uses Corollary~\ref{cor:potential_sum_costs}. Since the potential decreases during the execution of the algorithm, we get $\Phi_{D_{j+1}}(\state^j) \geq \Phi_{D_{j+1}}(\state^{j+1}) > \frac{n}{\constant} \cdot b_{j+1}$, which contradicts Lemma~\ref{lemma:potential_block}. Therefore $\playerCosts{\state^{j+1}} \leq \playerCosts{\state^j} + \frac{n(d+1)}{\constant} \cdot b_{j+1}$ and we use this  to show the lemma as follows,
\begin{align*}
	\playerCosts{\state^{m-1}}
	&\leq \playerCosts{\state^{m-2}} + \frac{n\cdot (d+1)}{\constant} \cdot b_{m-1}\\
	&\leq \playerCosts{\state^r} + \frac{n \cdot (d+1)}{\constant} \sum_{j=r+1}^{m-1}b_j\nonumber\\
	&= \playerCosts{\state^r} + \frac{n \cdot(d+1)}{\constant} \sum_{j=r+1}^{m-1}\playerCostsMax \cdot g^{-j}\nonumber\\
	&= \playerCosts{\state^r} + \frac{n\cdot(d+1)}{\constant} \sum_{j=r+1}^{m-1}b_r \cdot g^{r-j}\nonumber\\
	&\leq \playerCosts{\state^r} + \frac{n \cdot (d+1)}{\constant} \cdot 2  \cdot b_r \cdot g^{-1}\nonumber\\
	&\leq \playerCosts{\state^r} + \frac{2\cdot n \cdot (d+1)}{\constant \cdot  g} \cdot \playerCosts{\state^r}\nonumber\\
	&= \left( 1 + \frac{2\cdot n \cdot (d+1)}{\constant \cdot g} \right) \cdot \playerCosts{\state^r}
	= \left( 1 + \constant^2 \right) \cdot \playerCosts{\state^r}.
\end{align*}
\hfill 
\end{proof}
\vspace{3mm}

\begin{proof}[Proof of Lemma~\ref{lemma:costs_otherstrategy}]
Similarly to previous lemma, we first show  by contradiction the following. For two arbitrary successive phases $j$ and $j+1$ and an arbitrary alternative strategy $P'_i$ of player $i$,  $\playerCosts{\state_{-i}^{j+1}, \state'_i} \geq \playerCosts{\state_{-i}^j, \state'_i} - \frac{n\cdot (d+1)}{\constant}\cdot b_{j+1}$.
Thus, assume that  $\playerCosts{\state_{-i}^j, \state'_i} - \playerCosts{\state_{-i}^{j+1}, \state'_i} > \frac{n\cdot (d+1)}{\constant} \cdot b_{j+1}$. Since player $i$ does not deviate during phase $j+1$, the increase in her costs is caused by other players deviating to the resources she uses. Thus, there exists a set of resources $E' \subseteq E$ such that each resource in $E'$ is used by player $i$ and by at least one player in $D_{j+1}$ at state $\state^{j+1}$. Therefore 
\begin{align*}
	\sum_{e \in E'} \costshare{\state_{-i}^j, \state'_i} > \frac{n\cdot (d+1)}{\constant}\cdot b_{j+1} & \Rightarrow
	\sum_{e \in E'} \costshare{\state_{-i}^j, \state_i} > \frac{n \cdot (d+1)}{\constant}\cdot b_{j+1}.
\end{align*}
Following exactly the same steps as in proof of Lemma \ref{lemma:costs_end}, the previous yields to a  contradiction of Lemma \ref{lemma:potential_block}.
Thus,  $\playerCosts{\state_{-i}^{j+1}, \state'_i} \geq \playerCosts{\state_{-i}^j, \state'_i} - \frac{n \cdot (d+1)}{\constant} \cdot b_{j+1}$, which we use to show the lemma's statement as follows,
\begin{align*}
	\playerCosts{\state_{-i}^{m-1}, \state'_i}
	&\geq \playerCosts{\state_{-i}^{m-2}, \state'_i} - \frac{n\cdot(d+1)}{\constant}\cdot b_{m-1}\nonumber\\
	&\geq \playerCosts{\state_{-i}^{r}, \state'_i} - \frac{n \cdot (d+1)}{\constant}\cdot \sum_{j=r+1}^{m-1}b_j\nonumber\\
	&= \playerCosts{\state_{-i}^{r}, \state'_i} - \frac{n \cdot (d+1)}{\constant} \cdot \sum_{j=r+1}^{m-1}\playerCostsMax \cdot g^{-j}\nonumber\\
	&= \playerCosts{\state_{-i}^{r}, \state'_i} - \frac{n\cdot (d+1)}{\constant} \cdot \sum_{j=r+1}^{m-1}b_r \cdot g^{r-j}\nonumber\\
	&\geq \playerCosts{\state_{-i}^{r}, \state'_i} - \frac{n\cdot (d+1)}{\constant} \cdot 2  \cdot b_r \cdot g^{-1}\nonumber\\
	&\overset{b_r}{=} \playerCosts{\state_{-i}^{r}, \state'_i} - \frac{2 \cdot n \cdot (d+1)}{\constant \cdot g} \cdot \playerCosts{\state^r}\nonumber\\
	&\overset{g}{=} \playerCosts{\state_{-i}^{r}, \state'_i} - \constant^2 \cdot \playerCosts{\state^r}\nonumber\\
	&\overset{\gamma \leq \frac{1}{s}}{\geq} \playerCosts{\state_{-i}^{r}, \state'_i} - \frac{\constant}{s} \cdot \playerCosts{\state^r}\nonumber\\
	&\geq \playerCosts{\state_{-i}^{r}, \state'_i} - \constant \cdot \playerCosts{\state_{-i}^r, \state'}
	= \left(1-\constant\right) \cdot \playerCosts{\state_{-i}^{r}, \state'_i}.
\end{align*}
The second last inequality holds due to the $s$-approximate equilibrium for player $i$ in  $\state^r$.	
\hfill 
\end{proof}
\vspace{3mm}

\begin{proof}[Proof of Lemma~\ref{lemma:approximation_factor}]
Let $i$ be an arbitrary player who took her last move in phase $r$ and let $\state'_i$ be an arbitrary other strategy of player $i$. We use Lemma~\ref{lemma:costs_end} and Lemma~\ref{lemma:costs_otherstrategy} and the fact that player $i$ has no incentive to make a $s$-move in phase $r$ (by definition of the algorithm):
\begin{align}
	\frac{\playerCosts{\state^{m-1}}}{\playerCosts{\state_{-i}^{m-1}, \state'_i}}
	&\leq \frac{(1+\constant^2) \cdot \playerCosts{\statePhase}}{(1-\constant) \cdot \playerCosts{\state_{-i}^r, \state'_i}}\nonumber\\
	&\leq \left(\frac{1+\constant^2}{1-\constant}\right) \cdot \left(\frac{1}{\tStretch} -2\constant \right)^{-1}\nonumber\\
	&\leq \left(\frac{1+\constant^2}{1-\constant}\right) \cdot \left(\frac{1}{\tStretch} -2 \constant \right)^{-1}\nonumber
\end{align}

By minimizing the first part, we can get arbitrary close to $1$. For the second part, we need to fix a $\constant$ with $\constant < \frac{1}{2\tStretch}$. Therefore, the expression can be simplified to $\alpha = (1+O(\gamma)) \cdot \tStretch$.
\end{proof}
\vspace{3mm}

\begin{proof}[Proof of Lemma~\ref{alpha_order}] By Lemma ~\ref{lemma:approximation_factor} and Corolarry \ref{stretch:upperbound}, we get that  our main factor $\alpha$ (from Lemma~\ref{lemma:approximation_factor}) equals to 
$$ (1+O(\gamma)) \cdot \frac{(d+1)^2 \cdot (d+3)}{8} \cdot \frac{t \cdot (2^{\frac{1}{d+1}} - 1)^{-d}}{2^{-\frac{d}{d+1}} \cdot( 1+ t) - t},$$
where $\gamma$ is a small positive constant and $t=1+ \gamma$. Observe that factor $\alpha$ is essentially in the order of
$$\Theta(d^3) \cdot \left( \frac{1}{2^{\frac{1}{d+1}} -1}\right)^d.$$
We now claim that the order of the above is $\left(\frac{d}{\ln 2}\right)^d \cdot poly(d)$. To prove this, it is enough to show that $ \frac{1}{2^{\frac{1}{d+1}} -1} $  is assymptotically similar to $\frac{d}{\ln(2)}$.  
Applying L'Hospital's rule, this follows from the fact that
$$\lim_{d\to \infty} \frac{\frac{1}{d}}{2^{\frac{1}{d+1}} - 1} = 
\lim_{d\to \infty} \frac{- \frac{1}{d^2}}{- \frac{2^{\frac{1}{d+1}} \cdot \ln(2)}{(d+1)^2}} = \frac{1}{\ln(2)},$$
which completes the proof.
\hfill 
\end{proof}
\vspace{3mm}

\section{Proofs for the Sampling in Section~\ref{section:sampling}}

\begin{proof}[Proof of Lemma~\ref{lemma:sampling_shapley}]
The beginning of the proofs follows from the analysis in~\cite{DBLP:conf/cocoon/Liben-NowellSWW12}.
Let $X$ be the marginal contribution of player $i$ in a random permutation. Since $C_e$ is a polynomial of degree $d$ and monotone, we have $X \geq 0$. By the definition of the Shapley value, $\chi_{ie}(P) = E[X]$. By the definition of the cost functions, the maximum possible value of $X$ is achieved when $i$ is the last player in the ordering, this happens in $1/|S_e(P)|$ fraction of the permutations. $X$ achieves the maximum value with probability at least $1/|S_e(P)|$ and the maximum value is at most $|S_e(P)| \cdot \chi_{ie}(P)$ because of the expectation and the bounds of the values.

To upper bound the variance of $X$ we define a second random variable $Y$ which is $|S_e(P)| \cdot \chi_{ie}(P)$ with probability $1/n$ and $0$ otherwise. Then,
\begin{align*}
	Var(X) \leq Var(Y) = E[Y^2] - E[Y]^2
	= (|S_e(P)|-1) \cdot \chi_{ie}(P)^2 
\end{align*}
Since $\overline{MC}_{ie}(P) = \frac{1}{k} \sum_{j=1}^k MC_{ie}^j(P)$, $E[\overline{MC}_{ie}(P)] = E[X] = \chi_{ie}(P)$ and the single permutations are independent of each other, we get
$Var(\overline{MC}_{ie}(P)) = \frac{Var(X)}{k} \leq \frac{1}{k} (|S_e(P)|-1) \cdot \chi_{ie}(P)^2$.
Using Chebyshev's inequality, we get
\begin{align*}
	Pr[|\overline{MC}_{ie}(P) -  \chi_{ie}(P)| \geq \mu \chi_{ie}(P)]
	&\leq \frac{Var(\overline{MC}_{ie}(P))}{\chi_{ie}(P)^2 \mu^2}\\
	&\leq \frac{ (|S_e(P)|-1) \cdot \chi_{ie}(P)^2 }{k  \chi_{ie}(P)^2 \mu^2}
	= \frac{ |S_e(P)|-1}{k \cdot \mu^2}
\end{align*}

Let $k=\frac{4(|S_e(P)|-1)}{\mu^2}$, then $\overline{MC}_{ie}(P)$ is a $\mu$-approximation for $\chi_{ie}(P)$ with probability at least $3/4$.	
If we repeat this 
$$\log\left(2 n^{c+3} \cdot \max_{i \in N}\CP_i \cdot |E| \cdot \left(1+\log\left(\frac{\playerCostsMax}{\playerCostsMin}\right)\right) \cdot  (d+1) \cdot \gamma^{-9}\right)$$ times, using the median value of all runs and applying Chernoff bounds, we directly get a result with failure probability at most $$\frac{1}{n^c \cdot n \cdot \max_{i \in N}\CP_i \cdot |E| \cdot \left(1+\log\left(\frac{\playerCostsMax}{\playerCostsMin}\right)\right) \cdot 2 \cdot n^2 \cdot (d+1) \cdot \gamma^{-9}}.$$
\hfill 
\end{proof}
\vspace{3mm}

\begin{proof}[Proof of Lemma~\ref{lemma:sampling_improvement_step}]
The result follows directly by applying the union bound:
\begin{align*}
	&Pr[\exists i \in N:
	\exists P'_i \in \CP_i:
	\exists e \in P'_i :
	|\overline{MC}_{ie}(P_{-i}, P'_i) -  \chi_{ie}(P_{-i}, P'_i)|
	\geq \mu \cdot \chi_{ie}(P_{-i}, P'_i)]\\
	&\leq n \cdot \max_{i \in N}\CP_i \cdot |E| \cdot \frac{1}{n^c \cdot n \cdot \max_{i \in N}\CP_i \cdot |E| \cdot \left(1+\log\left(\frac{\playerCostsMax}{\playerCostsMin}\right)\right) \cdot 2 \cdot n^2 \cdot (d+1) \cdot \gamma^{-9}}\\
	&\leq \frac{1}{n^c \cdot \left(1+\log\left(\frac{\playerCostsMax}{\playerCostsMin}\right)\right) \cdot 2 \cdot n^2 \cdot (d+1) \cdot \gamma^{-9}}
\end{align*}
\hfill 
\end{proof}
\vspace{3mm}

\begin{proof}[Proof of Lemma~\ref{lemma:sampling_total}]
The result follows directly by applying the union bound:
\begin{align*}
	&Pr[\exists \text{ an improvement step in which the sampling fails}]\leq \\
 &\leq \frac{ \left(1+\log\left(\frac{\playerCostsMax}{\playerCostsMin}\right)\right) \cdot 2 \cdot n^2 \cdot (d+1) \cdot \gamma^{-9} }{n^c \cdot \left(1+\log\left(\frac{\playerCostsMax}{\playerCostsMin}\right)\right) \cdot 2 \cdot n^2 \cdot (d+1) \cdot \gamma^{-9}}\\
 & \leq \frac{1}{n^c}.
\end{align*}
\hfill 
\end{proof}
\vspace{3mm}

\end{document}